\documentclass[creativecommons]{eptcs}
\providecommand{\event}{E. Formenti (Ed.): AUTOMATA \& JAC 2012}
\usepackage{amssymb}
\usepackage{amsthm}

\newcommand\ZZ{\mathbb{Z}}

\title{Fixed Parameter Undecidability for Wang Tilesets}
\author{Emmanuel Jeandel\thanks{The first author is partially supported by ANR-09-BLAN-0164.
	Part of this work was done during a traineeship of the second author.}
\institute{
LIRMM, CNRS UMR 5506 - CC 477\\                                                
161 rue Ada,  34\,095 Montpellier Cedex 5, France}
\email{jeandel@lirmm.fr}\\
\and Nicolas Rolin
\institute{
LIP6\\
4 place Jussieu, 75252 Paris cedex 05, France}
\email{nicolas.rolin@ens-cachan.fr}
}
\def\authorrunning{E.~Jeandel \& N.~Rolin}
\def\titlerunning{Fixed Parameter Undecidability for Wang Tileset}

\usepackage{tikz}
\usetikzlibrary{automata}
\usetikzlibrary{patterns,snakes}
\newtheorem{theorem}{Theorem}[section]
\newtheorem{lemma}{Lemma}[section]
\newtheorem{proposition}{Proposition}[section]
\newtheorem{definition}{Definition}[section]
\newtheorem{corollary}{Corollary}[section]
\begin{document}
\maketitle
\begin{abstract}
Deciding if a given set of Wang tiles admits a tiling of the plane
is decidable if the number of Wang tiles (or the
number of colors) is bounded, for a trivial reason, as there are
only finitely many such tilesets. We prove however that the tiling
problem remains undecidable if the difference between the number of
tiles and the number of colors is bounded by 43.

One of the main new tool is the concept of Wang bars, which are
equivalently inflated Wang tiles or thin polyominoes.
\end{abstract}

\section*{Introduction}
Wang tiles are a model of computation introduced by Wang
\cite{wangpatternrecoII} to study decision procedures for some
fragments of first-order logic. The model is quite simple: We are
given a finite set of tiles, i.e. a squares with colored edges, and we
look at a way to tile the plane with the tiles so that contiguous
edges have the same color. Berger, a student of Wang, proved
\cite{BergerPhd} that the problem is algorithmically undecidable:
there is no way to decide, given a set of tiles, whether it can tile
the plane. One of the reasons for the difficulty of the problem is the
existence of \emph{aperiodic} sets of tiles, for which it is possible
to find a tiling, but no \emph{periodic} one.

If we happen to bound the number of different tiles, say by $100$,
then the problem becomes trivially decidable, as there are only a
finite number of such sets of tiles. The same is true is we bound the
number of colors. We will try here to find a good parameter, for which
this trivial situation does not happen: We prove in this article that the problem remains undecidable if the
difference between the number of tiles and the number of colors is
bounded by a constant, here $43$. Furthermore, we prove that if the
difference is small enough, then the tiling problem is actually
decidable, because there are no aperiodic tilesets with such a
difference. In another way, if there are too many colors (compared to
the number of tiles), the set of tiles cannot be aperiodic.

This last result has pragmatic implications that motivated this
research. The first author is at present conducting experiments
to find an aperiodic set of tiles with as few tiles as
possible (actually set of tiles that are \emph{candidates} for being
aperiodic, as the problem is not decidable{\ldots}). The main
bottleneck for this kind of approach is indeed the number of such sets of
tiles, thus proving that we need a small number of colors to be aperiodic
may critically reduce the size of the search space.

The article is organised as follows. We introduce the relevant
definitions in the first section. We then prove that the problem under
consideration is equivalent to the question of tileability with a
fixed number of \emph{Wang bars}, which are an intermediate between Wang
tiles and polyominoes. We then proceed to prove the upper and lower
bound on the parameter that is needed to have an undecidability result.

\newcommand\wang[4]{
        \draw[line width=0.1pt, fill=black!30] (0,0) -- +(2,2) -- +(0,4) -- +(0,0);
        \draw[line width=0.1pt] (0,0) -- +(2,2) -- +(4,0) -- +(0,0);
\begin{scope}[xshift=4cm, yshift=4cm]
        \draw[line width=0.1pt,fill=black!30] (0,0) -- +(-2,-2) -- +(0,-4) -- +(0,0); 
        \draw[line width=0.1pt] (0,0) -- +(-2,-2) -- +(-4,0) -- +(0,0);
\end{scope}%
                \draw (1,2) node {#1};
                \draw (3,2) node {#2};
                \draw (2,3.5) node {#3};
                \draw (2,.5) node {#4};
}

\section{Definitions}
\subsection{Wang tiles}
A Wang tile $t$ is a square tile with colored edges, as represented in
Fig.\ref{fig:wangtile}. Formally, it is given by a quadruplet $(t_e, t_w, t_n,
t_s)$ of symbols, called colors.
A \emph{tileset} $\tau$ is a finite set of Wang tiles.
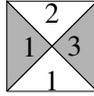
\begin{figure}
\begin{center}
	\begin{tikzpicture}[scale=0.3]
\begin{scope}[xshift=0 cm,yshift=0 cm]
\wang 1 3 2 1
\end{scope}
\end{tikzpicture}
\end{center}
\caption{A Wang tile}
\label{fig:wangtile}
\end{figure}
A \emph{tiling} of the plane by $\tau$ is a map $f$ from the discrete
plane $\ZZ^2$ to $\tau$ so that two tiles that share a common edge
agree on the color: For all integers $i,j$, we have $f(i,j)_e = f(i+1,j)_w$ and $f(i,j)_n = f(i,j+1)_s$.

The main question we study in this article is the Domino Problem: To
decide, given a tileset $\tau$, if there is a tiling by $\tau$.
This problem was proven undecidable by Berger \cite{BergerPhd}.

A reason for the complexity of the problem is the existence of
\emph{aperiodic} tilesets. A tiling $t$ is \emph{periodic}
if there exists $p$ so that $t(i,j) = t(i,j+p) = t(i+j,p)$ for all $i,j$.
If our tileset $\tau$ admits a periodic tiling, this tiling is easy to
find, by just testing all possible finite maps from $[0,n]^2$ to
$\tau$. 

A tileset is said to be periodic if it admits a periodic tiling. It
can be proven equivalent to the existence of a tiling $t$ which is only
horizontally periodic, that is $t(i,j) = t(i+p,j)$ for all $i,j$ and
for some $p$.
However there exist \emph{aperiodic} tilesets, that is tilesets
that tile the plane, but that admit no periodic tiling.
Such a tileset, due to Culik \cite{Culik} based on work by Kari
\cite{Kari14} is depicted in Fig.~\ref{fig:kariculik}
\begin{figure}[htbp]
	\begin{tikzpicture}[scale=0.3]

\begin{scope}[xshift=0 cm,yshift=0 cm]
\wang 1 2 2 3
\end{scope}
\begin{scope}[xshift=0 cm,yshift=6 cm]
\wang 1 3 2 2
\end{scope}
\begin{scope}[xshift=6 cm,yshift=0 cm]
\wang 2 3 2 3
\end{scope}
\begin{scope}[xshift=6 cm,yshift=6 cm]
\wang 2 1 1 2
\end{scope}
\begin{scope}[xshift=12 cm,yshift=0 cm]
\wang 3 1 1 3
\end{scope}
\begin{scope}[xshift=12 cm,yshift=6 cm]
\wang 3 2 1 2
\end{scope}
\begin{scope}[xshift=18 cm,yshift=0 cm]
\wang 4 4 4 1
\end{scope}
\begin{scope}[xshift=18 cm,yshift=6 cm]
\wang 4 4 3 2
\end{scope}
\begin{scope}[xshift=24 cm,yshift=0 cm]
\wang 4 5 2 1
\end{scope}
\begin{scope}[xshift=24 cm,yshift=6 cm]
\wang 4 5 2 4
\end{scope}
\begin{scope}[xshift=30 cm,yshift=0 cm]
\wang 5 5 4 1
\end{scope}
\begin{scope}[xshift=30 cm,yshift=6 cm]
\wang 5 5 3 2
\end{scope}
\begin{scope}[xshift=36 cm,yshift=0 cm]
\wang 5 4 2 2
\end{scope}

\end{tikzpicture}
\caption{An aperiodic tileset with 13 tiles from Kari and Culik}
\label{fig:kariculik}
\end{figure}
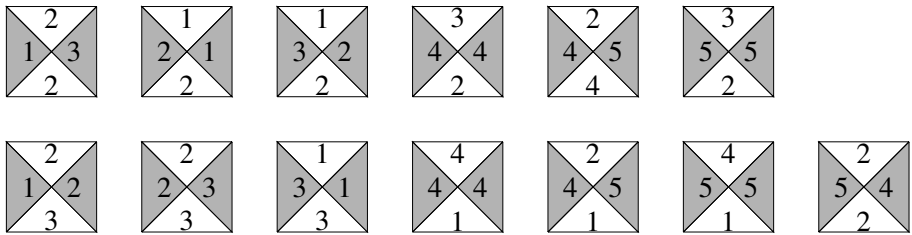
However, every such tileset admits tilings with some specific
regularity properties. We will use the following lemma in some of
the proofs:
\begin{lemma}\label{lemma:twice}
  Let $\tau$ be a tileset that tiles the plane. Then there exists a
  tiling $f$ by $\tau$ with the following property: For every tile
  $t$ that appears in $f$, there exist a row where $t$ appears at least twice
\end{lemma}
We can prove much more (for example the existence of a quasiperiodic
configuration \cite{Du}), but this will be sufficient for this article.
\begin{proof}
	Take a tiling $f$ with the minimal number of different tiles.
	If some tile $t$ appears at most once on each row of $f$, then we
	can find in $f$ big squares where $t$ does not appear. By a
	compactness argument, we can extract from $f$ a tiling $f'$
	without $t$, hence with strictly less different tiles, a
	contradiction.
\end{proof}	

\subsection{The graph approach}
If we look at tilings of only one row by $\tau$, a good way to
understand the tilings is using a (labeled) graph: Represent each
color by a vertex, and for each Wang tile $t$ with west color $c$ and
east color $c'$, add an edge from $c$ to $c'$ labeled with its north
and south color. An example corresponding to Fig.\ref{fig:kariculik} is provided in
Fig.\ref{fig:karigraph}. Note that for this particular example the graph has two
connected components, that are strongly connected (the relevant
vocabulary on graph theory will be defined below).

\begin{figure}	

\begin{tikzpicture}[>=latex,line join=bevel,]
  \pgfsetlinewidth{1bp}
\pgfsetcolor{black}
  \draw [->] (13.721bp,92.025bp) .. controls (6.0943bp,80.801bp) and (-0.59958bp,65.753bp)  .. (7bp,54bp) .. controls (16.725bp,38.96bp) and (34.395bp,30.141bp)  .. (60.522bp,22.254bp);
  \definecolor{strokecol}{rgb}{0.0,0.0,0.0};
  \pgfsetstrokecolor{strokecol}
  \draw (24bp,63bp) node {(2,3)};
  \draw [->] (298.35bp,202.11bp) .. controls (308.53bp,202.29bp) and (317bp,200.92bp)  .. (317bp,198bp) .. controls (317bp,196.13bp) and (313.53bp,194.89bp)  .. (298.35bp,193.89bp);
  \draw (334bp,198bp) node {(4,1)};
  \draw [->] (292.03bp,210.25bp) .. controls (317.35bp,219.22bp) and (351bp,215.14bp)  .. (351bp,198bp) .. controls (351bp,183bp) and (325.23bp,178bp)  .. (292.03bp,185.75bp);
  \draw (368bp,198bp) node {(3,2)};
  \draw [->] (38.138bp,124.71bp) .. controls (47.17bp,138.25bp) and (60.038bp,157.56bp)  .. (75.995bp,181.49bp);
  \draw (79bp,153bp) node {(1,2)};
  \draw [->] (298.35bp,112.11bp) .. controls (308.53bp,112.29bp) and (317bp,110.92bp)  .. (317bp,108bp) .. controls (317bp,106.13bp) and (313.53bp,104.89bp)  .. (298.35bp,103.89bp);
  \draw (334bp,108bp) node {(4,1)};
  \draw [->] (292.03bp,120.25bp) .. controls (317.35bp,129.22bp) and (351bp,125.14bp)  .. (351bp,108bp) .. controls (351bp,93.004bp) and (325.23bp,88.002bp)  .. (292.03bp,95.749bp);
  \draw (368bp,108bp) node {(3,2)};
  \draw [->] (248.17bp,189.19bp) .. controls (236.22bp,183.42bp) and (222.87bp,174.58bp)  .. (216bp,162bp) .. controls (207.35bp,146.15bp) and (223.01bp,132.13bp)  .. (248.79bp,117.63bp);
  \draw (233bp,153bp) node {(2,1)};
  \draw [->] (272bp,179.79bp) .. controls (272bp,167.34bp) and (272bp,150.61bp)  .. (272bp,126.19bp);
  \draw (289bp,153bp) node {(2,4)};
  \draw [->] (102.78bp,32.838bp) .. controls (115.86bp,46.311bp) and (133.49bp,67.469bp)  .. (141bp,90bp) .. controls (146.06bp,105.18bp) and (146.06bp,110.82bp)  .. (141bp,126bp) .. controls (134.78bp,144.66bp) and (121.62bp,162.38bp)  .. (102.78bp,183.16bp);
  \draw (161bp,108bp) node {(1,3)};
  \draw [->] (75.862bp,34.708bp) .. controls (66.83bp,48.254bp) and (53.962bp,67.557bp)  .. (38.005bp,91.493bp);
  \draw (79bp,63bp) node {(1,2)};
  \draw [->] (290.02bp,121.44bp) .. controls (296.68bp,127.51bp) and (303.46bp,135.28bp)  .. (307bp,144bp) .. controls (310.01bp,151.41bp) and (310.01bp,154.59bp)  .. (307bp,162bp) .. controls (304.73bp,167.59bp) and (301.13bp,172.79bp)  .. (290.02bp,184.56bp);
  \draw (326bp,153bp) node {(2,2)};
  \draw [->] (60.522bp,193.75bp) .. controls (42.073bp,189.25bp) and (18.743bp,180.16bp)  .. (7bp,162bp) .. controls (1.2409bp,153.09bp) and (3.6903bp,142.29bp)  .. (13.721bp,123.97bp);
  \draw (24bp,153bp) node {(2,3)};
  \draw [->] (93.55bp,180.3bp) .. controls (95.368bp,174.54bp) and (97.092bp,168.08bp)  .. (98bp,162bp) .. controls (105.09bp,114.53bp) and (105.09bp,101.47bp)  .. (98bp,54bp) .. controls (97.589bp,51.247bp) and (97.01bp,48.413bp)  .. (93.55bp,35.704bp);
  \draw (120bp,108bp) node {(2,2)};
\begin{scope}
  \definecolor{strokecol}{rgb}{0.0,0.0,0.0};
  \pgfsetstrokecolor{strokecol}
  \draw (87bp,198bp) ellipse (27bp and 18bp);
  \draw (87bp,198bp) node {1};
\end{scope}
\begin{scope}
  \definecolor{strokecol}{rgb}{0.0,0.0,0.0};
  \pgfsetstrokecolor{strokecol}
  \draw (87bp,18bp) ellipse (27bp and 18bp);
  \draw (87bp,18bp) node {3};
\end{scope}
\begin{scope}
  \definecolor{strokecol}{rgb}{0.0,0.0,0.0};
  \pgfsetstrokecolor{strokecol}
  \draw (27bp,108bp) ellipse (27bp and 18bp);
  \draw (27bp,108bp) node {2};
\end{scope}
\begin{scope}
  \definecolor{strokecol}{rgb}{0.0,0.0,0.0};
  \pgfsetstrokecolor{strokecol}
  \draw (272bp,108bp) ellipse (27bp and 18bp);
  \draw (272bp,108bp) node {5};
\end{scope}
\begin{scope}
  \definecolor{strokecol}{rgb}{0.0,0.0,0.0};
  \pgfsetstrokecolor{strokecol}
  \draw (272bp,198bp) ellipse (27bp and 18bp);
  \draw (272bp,198bp) node {4};
\end{scope}
\end{tikzpicture}

\caption{The tileset of Kari and Culik represented as a labeled graph.}
\label{fig:karigraph}
\end{figure}
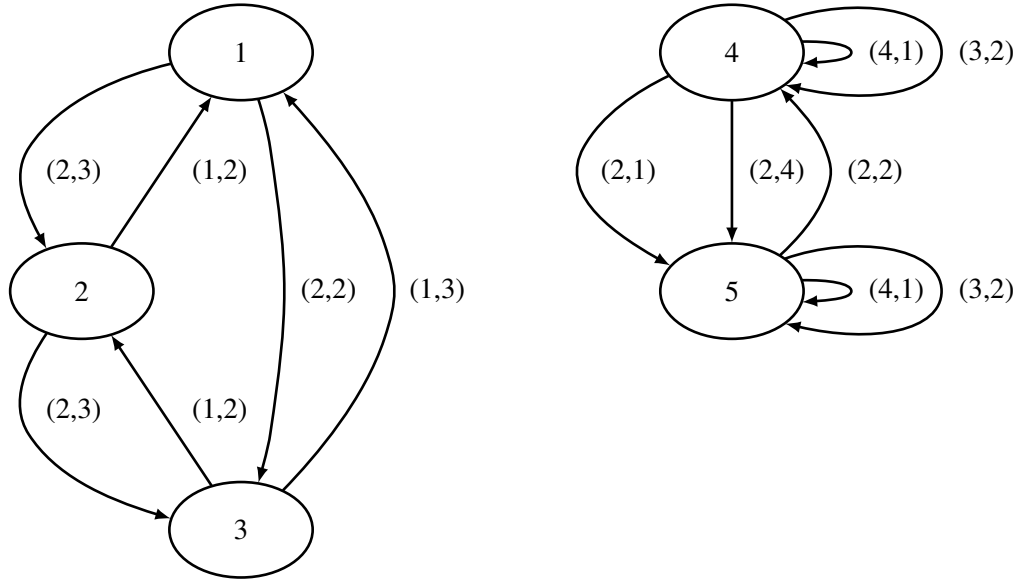
It is quite clear that a tiling of a row is equivalent to a biinfinite
path in this graph. We will now describe more precisely the connection.

First, some definitions. A labeled (multi)graph where each edge is labeled
by a pair of symbols  will be called a \emph{$1$-labeled graph}.

Given a $1$-labeled graph $G$, a pair of biinfinite words
$(u^1,u^2)$ is
\emph{compatible} if $(u^1,u^2)$ correspond to a biinfinite path on the
graph. That is, there exists vertices $(v_i)_{i \in \mathbb{Z}}$ so
that for all $i \in \ZZ$ there is an edge from $v_i$ to $v_{i+1}$
labeled $(u^1_i, u^2_i)$.
A \emph{tiling by $G$} is a biinfinite sequences of biinfinite
words $(u^j)_{j \in \ZZ}$ so that $(u^j, u^{j+1})$ is compatible for
all $j$.

Now it is clear that, upto notations, a 1-labeled graph $G$ with $n$ edges and
$k$ vertices is exactly the same as a set $\tau$ of $n$ Wang tiles with 
$k$ colors on its east/west wide, and that a tiling by $G$ exists if
and only if a tiling by $\tau$ exists.

\subsection{The parameter}
For a given set of Wang tiles $\tau$, we denote by $n(\tau)$ (or
simply $n$) the number of tiles of $\tau$ and by $c(\tau)$ (or simply
$c$) the maximum over all four sides of the number of colors.
By definition, it is clear that $c(\tau) \leq n(\tau)$.

\textbf{\emph{By rotating the Wang tiles, we may suppose w.lo.g.,
	and we do from now, that the maximum number of colors is reached
	on the west side.}}

If we view the set of Wang tiles as a 1-labeled graph $G$, $n(\tau)$
represents the number of edges of $G$, and $c(\tau)$ the number of
vertices of $G$ (more precisely the number of vertices with non-zero
outdegree).

It is obvious that $n(\tau)$ and $c(\tau)$ are not good parameters, in
the sense that the Domino problem is trivially decidable for a fixed
$n$ (resp. $c$), as there are only finitely many tilesets with $n$
tiles (resp. using at most $c$ colors on each side).

However we now observe that $n(\tau)-c(\tau)$ might be a good parameter.
Indeed, if the number of colors on the west side is roughly the same
as the number of tiles, then this means that for almost every tile, there
usually will be at most one tile that can be put on its right.
The
tiling problem becomes highly constrained and we might expect the
problem to be easier. This is indeed the case:

\begin{theorem}
	\label{thm:minima}
	The domino problem for tilesets of parameter $n-c \leq 1$
	is decidable. More precisely, a tileset with parameter $n-c \leq
	1$ admits a tiling if and only if it admits a periodic tiling.	
\end{theorem}	

On the other hand, if $n-c$ is big, we have more choices for the next
tile to put. We will prove:
\begin{theorem}
	\label{thm:maxima}
	The domino problem for tilesets of parameter $n - c = 43$ is
	undecidable. It remains undecidable for higher parameters.
\end{theorem}	

What happens between $1$ and $43$ is unknown. What is clear is
that the situation becomes more complex very quickly:
\begin{proposition}
	There exists a tileset $\tau$ of parameter $n-c = 8$ that admits a
	tiling but no periodic tiling (i.e. that is aperiodic).
\end{proposition}	
This is exactly the tileset of Fig.\ref{fig:kariculik}.

The rest of the paper will be organized as follows.
In the next section, we will introduce a new object, called Wang bars,
and prove that tileability with a tileset of parameter $n-c = k$
is (somewhat) equivalent to tileability with $O(k)$ Wang bars. Then we
will prove in the following sections that: (a) tileabilty with 44
Wang bars is undecidable (thus proving Theorem\ref{thm:maxima}) and that (b)
tileability with 2 Wang bars is decidable (thus proving Theorem\ref{thm:minima}).

\section{Wang bars}

A Wang bar is a Wang tile which is bigger
horizontally.  Formally, a Wang bar $b$ is a quadruplet $(b_e, b_w,
b_n, b_s)$ where $b_n$ and $b_s$ are words of the same length over
some alphabet $C$. We denote by $|b|$ the length of the word $b_n$.
Letters of $b_n$ are numeroted from $1$ to $|b|$. It is clear from the
definition that a Wang tile is a Wang bar of length $1$.

A \emph{barset} is a set of Wang bars. See Fig.~\ref{fig:barwang} for
an example.

A tiling of the plane by Wang bars is a partition of the plane by Wang
bars so that consecutive Wang bars have the same colors on their shared
edge.
Here is a formal yet nonintuitive definition.

\begin{definition}
Let $B$ be a barset.
A tiling of the plane by $B$ is a pair of a map $f$ from $\mathbb{Z}^2$ to
$B$ ($f(i,j)$ is called the bar at $(i,j)$), and a map $p$ from
$\mathbb{Z}^2$ to $\mathbb{N}$ ($p(i,j)$ is the \emph{position} inside
the bar $(i,j)$) so that:
\begin{itemize}
	\item $0 < p(i,j) \leq |f(i,j)|$ (the position inside the bar is less
	  than the length of the bar)
	\item If $p(i,j) < |f(i,j)|$ then $p(i+1,j) = p(i,j)+1$ and
	  $f(i+1,j) = f(i,j)$ (if we are inside the bar, this is still the
	  same bar)
	\item If $p(i,j)=|f(i,j)|$ (we are at the end of the bar) 
	  then $p(i+1,j) = 1$ (a new bar starts) and $f(i,j)_e =
	  f(i+1,j)_w$ (horizontal colors match)
	\item $(f(i,j)_n)_{p(i,j)} = (f(i,j+1)_s)_{p(i,j+1)}$ (vertical colors match)
\end{itemize}	
\end{definition}

\newcommand\wbarre[4]{
        \draw[line width=0.1pt] (0,0) -- +(2,2) -- +(4,0) -- +(0,0);
                \draw (1,2) node {#1};
                \draw (3,2) node {#2};
                \draw (2,3.5) node {#3};
                \draw (2,.5) node {#4};
}
\newcommand\bara{
\draw(0,0) rectangle (12,4);
\draw[line width=0.1pt, fill=black!30] (0,0) -- +(2,2) -- +(0,4) -- +(0,0);
\draw[line width=0.1pt, fill=black!30] (12,0) -- +(-2,2) -- +(0,4) -- +(0,0);
\draw (1,2) node {0};
\draw (11,2) node {0};
\draw (2,1) node {0};
\draw (6,1) node {1};
\draw (10,1) node {0};
\draw (2,3) node {1};
\draw (6,3) node {0};
\draw (10,3) node {2};
}

\newcommand\barb{
\draw(0,0) rectangle (8,4);
\draw[line width=0.1pt, fill=black!30] (0,0) -- +(2,2) -- +(0,4) -- +(0,0);
\draw[line width=0.1pt, fill=black!30] (8,0) -- +(-2,2) -- +(0,4) -- +(0,0);
\draw (1,2) node {0};
\draw (7,2) node {1};
\draw (2,1) node {2};
\draw (6,1) node {1};
\draw (2,3) node {1};
\draw (6,3) node {0};
}

\newcommand\barc{
\draw(0,0) rectangle (16,4);
\draw[line width=0.1pt, fill=black!30] (0,0) -- +(2,2) -- +(0,4) -- +(0,0);
\draw[line width=0.1pt, fill=black!30] (16,0) -- +(-2,2) -- +(0,4) -- +(0,0);
\draw (1,2) node {1};
\draw (15,2) node {0};
\draw (2,1) node {0};
\draw (6,1) node {2};
\draw (10,1) node {1};
\draw (14,1) node {0};
\draw (2,3) node {0};
\draw (6,3) node {1};
\draw (10,3) node {0};
\draw (14,3) node {0};
}
\begin{figure}[htbp]
\begin{center}
\begin{tikzpicture}[scale=0.3]
	\bara
\begin{scope}[xshift=16cm]
	\barb
\end{scope}	
\begin{scope}[yshift=-6cm, xshift=4cm]
	\barc
\end{scope}	
\end{tikzpicture}
\end{center}
\caption{A set of three Wang bars. The first one is formally defined
  as $(0,0,102,010)$.}
\label{fig:barwang}
\end{figure}
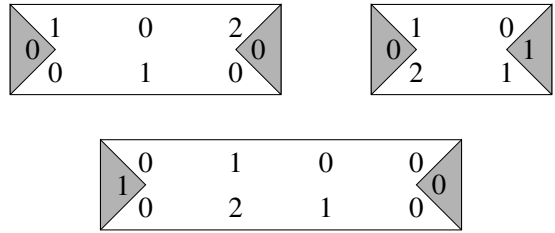

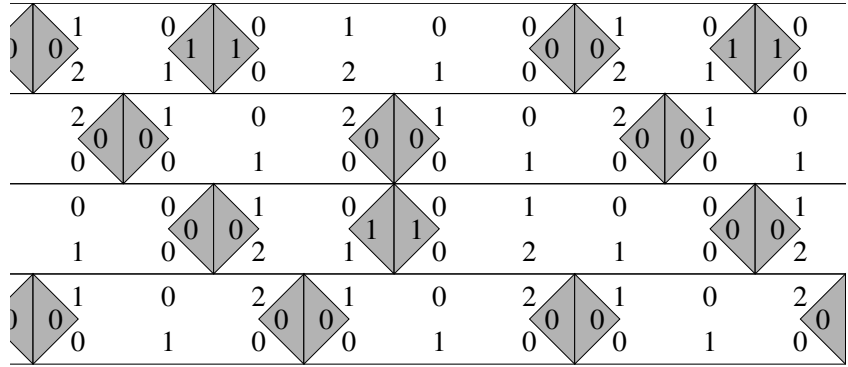
\begin{figure}[htb]
\begin{center}
\begin{tikzpicture}[scale=0.3]
\clip (-1,-1) rectangle (36,16);
\begin{scope}[xshift=-12cm]	\bara \end{scope}
\begin{scope}[xshift=0cm]	\bara \end{scope}
\begin{scope}[xshift=12cm]	\bara \end{scope}
\begin{scope}[xshift=24cm]	\bara \end{scope}
\begin{scope}[xshift=-8cm,yshift=4cm]	\barc \end{scope}
\begin{scope}[xshift=8cm,yshift=4cm]	\barb \end{scope}
\begin{scope}[xshift=16cm,yshift=4cm]	\barc \end{scope}
\begin{scope}[xshift=32cm,yshift=4cm]	\barb \end{scope}
\begin{scope}[xshift=-8cm,yshift=8cm]	\bara \end{scope}
\begin{scope}[xshift=4cm,yshift=8cm]	\bara \end{scope}
\begin{scope}[xshift=16cm,yshift=8cm]	\bara \end{scope}
\begin{scope}[xshift=28cm,yshift=8cm]	\bara \end{scope}
\begin{scope}[xshift=-16cm,yshift=12cm]	\barc \end{scope}
\begin{scope}[xshift=0cm,yshift=12cm]	\barb \end{scope}
\begin{scope}[xshift=8cm,yshift=12cm]	\barc \end{scope}
\begin{scope}[xshift=24cm,yshift=12cm]	\barb \end{scope}
\begin{scope}[xshift=32cm,yshift=12cm]	\barc \end{scope}
\end{tikzpicture}
\end{center}
\caption{A fragment of a tiling by the three bars of
  Figure\ref{fig:barwang}. The
  acute reader may try to convince himself that there exists a periodic
  tiling.}
	\label{fig:bartile}
\end{figure}
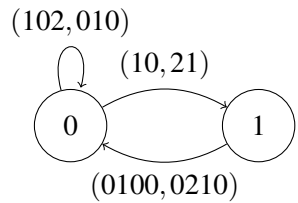
\begin{figure}[htbp]
\begin{center}	
\begin{tikzpicture}[->, node distance=2.5cm]
\node[state] (A) {0};
\node[state] (B) [right of=A] {1};
\path (A) edge[bend left] node[above] {$(10,21)$} (B);
\path (B) edge[bend left] node[below] {$(0100,0210)$} (A);
\path (A) edge[loop above] node[above] {$(102,010)$} (A);
\end{tikzpicture}
\end{center}	
	\caption{The example of Fig. \ref{fig:barwang} in the graph formalism}
	\label{fig:bargraph}
\end{figure}

An example is provided in Fig.\ref{fig:bartile}.

Wang bars have of course a graph counterpart. We define a
\emph{$\omega$-labeled graph} to be a directed (multi)graph
where every edge is labeled by a pair of words of the same length.
We do not define explicitely how a $\omega$-labeled graph tiles the
plane, but it should be clear. An example of such a graph is provided
in Fig. \ref{fig:bargraph}

It is quite clear that tileability with Wang bars can be reduced to
tileability with Wang tiles: Just consider a Wang bar of length $p$ as
$p$ Wang tiles.

\begin{theorem}
	Tileability with $k$ bars is many-one reducible to tileability of Wang tiles
	with parameter at most $k-1$. That is, every barset $B$ with $k$ bars
	can be transformed into a tileset $W$ with parameter at most  $k-1$
	so that $B$ tiles the plan iff $W$ tiles the plane.
\end{theorem}
\begin{proof}
	Let $c$ be the number of colors that appear in the east side of a
	Wang bar of $B$ and $k$ be the number of bars.
	
	We now split each Wang bar into Wang tiles, introducing new colors
	each time. If we see a Wang tile as a Wang bar, it is easy to see
	that each time we split a Wang bar into $p$ Wang tiles(bars), we
	increase the number of bars by $p-1$ and the number of colors by
	$p-1$, hence the difference between the number of bars and the
	number of colors on the east side stays constant.
	   
    As a consequence, the tileset $\tau$ we obtain at the end 
	has $n(\tau) = p+k$ tiles and $p+c$ colors  in the
	east side (so that $c(\tau) \geq p+c$), hence is of parameter
	$n(\tau) - c(\tau) \leq k-c \leq k-1$.
\end{proof}	

The converse is less clear. We will use the graph formalism to
proceed. We start from a tileset with parameter $n-c = k$. We suppose
w.l.o.g. that the maximum number of colors is reached on the west
side.  In the graph formalism, this means that $G$ has $n$ edges, and
$c$ vertices with at least one outgoing edge.

First we recall some definitions from digraph theory. 
The \emph{outdegree} (resp. \emph{indegree}) 
of a vertex $v$ is the number of labeled edges that start from $v$
(resp. end in $v$)

A (directed) \emph{path} in $G$ from $u$ to $v$ 
is a sequence $u = u_1 \dots u_p = v$ (possibly with $p = 1$) of vertices so that for
all $i < p$ there exists an edge from $u_i$ to $u_{i+1}$.

\begin{proposition}
	\label{prop:etapeun}
	Let $\tau$ be a tileset represented as a graph $G$.
	Suppose that there exists two vertices $u$, $v$ so that:
	\begin{itemize}
		\item There exists an edge from $u$ to $v$, corresponding to
		  the tile $t$
		\item There exists no path from $v$ to $u$.
	\end{itemize}		
	Then $\tau$ tiles the plane if and only if $\tau - \{t\}$ tiles
	the plane.
\end{proposition}	
\begin{proof}
One direction is clear. Suppose that $\tau$ tiles the plane. 
We know from Lemma \ref{lemma:twice} that there exists a tiling $f$, where each tile 
that appears in $f$ appears twice in some row.  $t$ cannot appear
in such a tiling, as it is not possible to go back from 
its east side (the color $v$) to the west side (the color $u$). Hence
the tiling $f$ does not use $t$ and $\tau - \{t\}$ tiles the plane.
\end{proof}	

We can now suppose w.l.o.g that there is no such tile $t$ in $\tau$.
In particular, in the graph, every vertex is of outdegree (resp.
indegree) at least one. Moreover, by a straightforward induction, this
implies that if there is a path from $u$ to $v$, then there is a path
from $v$ to $u$ (In digraph terms, all connected components are
strongly connected).

For two vertices $u$ and $v$, we say that $u \sim_G v$ if there exists a
path from $u$ to $v$. $\sim_G$ is in our case an equivalence relation (it is reflexive because every vertex is of
outdegree at least one). Note that on each row of a tiling by $\tau$, the
horizontal colors belong to the same equivalence class.

These equivalence classes are usually called strongly connected
components in the theory of digraphs.

\begin{proposition}
	\label{prop:etapedeux}
	Let $\tau$ be a tileset represented as a graph $G$ for which
	$\sim_G$ is an equivalence relation.

    Let $H$ be an equivalence class so that every vertex in $H$ is of
	outdegree exactly one, and $T$ the tiles corresponding to $H$.
	
	Then:	    
	\begin{itemize}
		\item Either $\tau - T$ tiles the plane
		\item Or there is a periodic tiling by $\tau$
		\item Or $\tau$ does not tile the plane
	\end{itemize}
\end{proposition}
\begin{proof}
Suppose that $\tau$ tiles the plane. By Lemma \ref{lemma:twice}, there exists a
tiling where each tile $t$ that appears in $f$ appears in at least two
different rows (as it appears twice on some column{\ldots}). If no tile of $T$ appear in $f$ we are done.
Otherwise, notice that each row where a tile from $T$ appears must be
tiled periodically, of period $|T|$.
Now, there exist two rows with a tile from $T$, say $0$ and $p > 0$.
We can now obtain a new tiling $f'$ which is periodic vertically
by $f'(i,j) = f(i,j \mathop{\mathrm{mod}} p)$. But this 
implies that there exists a tiling which is periodic.
\end{proof}	

We now are ready to prove the theorem.

\begin{theorem}\label{thm:wangtobar}
	Tileability of tilesets with parameter at most $k$ is reducible
	to tileability with at most $2k$ bars.
	More precisely, using tileability with at most $2k$ bars as on
	oracle, we can design an algorithm that, given a tileset $W$
	with parameter at most $k$, decides whether $W$ tiles the plane,
	asking at most one question to the oracle.
\end{theorem}	
Note: The reduction we have is \textbf{not} a many-one reduction. It
is usually called a weak-truth-table reduction.

\begin{figure}[htbp]
\begin{tabular}{|c|c|}
\hline	
	(I) & 
\begin{tikzpicture}[scale=0.3]
	\begin{scope}[xshift=0 cm,yshift=0 cm]
\wang 3 1 2 3
\end{scope}
\begin{scope}[xshift=0 cm,yshift=6 cm]
\wang 1 3 2 2
\end{scope}
\begin{scope}[xshift=6 cm,yshift=0 cm]
\wang 5 3 2 3
\end{scope}
\begin{scope}[xshift=6 cm,yshift=6 cm]
\wang 6 2 1 2
\end{scope}
\begin{scope}[xshift=12 cm,yshift=0 cm]
\wang 4 2 1 3
\end{scope}
\begin{scope}[xshift=12 cm,yshift=6 cm]
\wang 4 6 1 2
\end{scope}
\begin{scope}[xshift=18 cm,yshift=0 cm]
\wang 2 4 4 1
\end{scope}
\begin{scope}[xshift=18 cm,yshift=6 cm]
\wang 4 7 3 2
\end{scope}
\begin{scope}[xshift=24 cm,yshift=0 cm]
\wang 8 7 2 1
\end{scope}
\begin{scope}[xshift=24 cm,yshift=6 cm]
\wang 9 8 2 4
\end{scope}
\begin{scope}[xshift=30 cm,yshift=0 cm]
\wang 9 8 4 1
\end{scope}
\begin{scope}[xshift=30 cm,yshift=6 cm]
\wang 7 9 3 2
\end{scope}
\end{tikzpicture}
\\
\hline
(1) & 
\begin{tikzpicture}[->, node distance=2.5cm]
\node[state] (A) {1};
\node[state] (C) [right of=A] {3};
\node[state] (B) [right of=C] {2};
\node[state] (D) [right of=B] {4};
\node[state] (E) [below of=A] {5};
\node[state] (G) [right of=D] {7};
\node[state] (F) [below of=B] {6};
\node[state] (H) [below of=G] {8};
\node[state] (I) [right of=H] {9};
\path (C) edge[bend left] node[below] {$(2,3)$} (A);
\path (A) edge[bend left] node[above] {$(2,2)$} (C);
\path (E) edge node[below right] {$(2,3)$} (C);
\path (F) edge node[left] {$(1,2)$} (B);
\path (D) edge[bend left] node[below] {$(1,3)$} (B);
\path (D) edge node[right] {$(1,2)$} (F);
\path (B) edge[bend left] node[above] {$(4,1)$} (D);
\path (D) edge node[above] {$(3,2)$} (G);
\path (H) edge node[left] {$(2,1)$} (G);
\path (I) edge[bend left] node[below] {$(2,4)$} (H);
\path (I) edge[bend right] node[above] {$(4,1)$} (H);
\path (G) edge node[above right] {$(3,2)$} (I);
\end{tikzpicture}
\\
\hline

(2) & 
\begin{tikzpicture}[->, node distance=2.5cm]
\node[state] (A) {1};
\node[state] (C) [right of=A] {3};
\node[state] (B) [right of=C] {2};
\node[state] (D) [right of=B] {4};
\node[state] (G) [right of=D] {7};
\node[state] (F) [below of=B] {6};
\node[state] (H) [below of=G] {8};
\node[state] (I) [right of=H] {9};
\path (C) edge[bend left] node[below] {$(2,3)$} (A);
\path (A) edge[bend left] node[above] {$(2,2)$} (C);
\path (F) edge node[left] {$(1,2)$} (B);
\path (D) edge[bend left] node[below] {$(1,3)$} (B);
\path (D) edge node[right] {$(1,2)$} (F);
\path (B) edge[bend left] node[above] {$(4,1)$} (D);
\path (H) edge node[left] {$(2,1)$} (G);
\path (I) edge[bend left] node[below] {$(2,4)$} (H);
\path (I) edge[bend right] node[above] {$(4,1)$} (H);
\path (G) edge node[above right] {$(3,2)$} (I);
\end{tikzpicture}
\\
\hline
(3) &  
\begin{tikzpicture}[->, node distance=2.5cm]
\node[state] (B)  {2};
\node[state] (D) [right of=B] {4};
\node[state] (G) [right of=D] {7};
\node[state] (F) [below of=B] {6};
\node[state] (H) [below of=G] {8};
\node[state] (I) [right of=H] {9};
\path (F) edge node[left] {$(1,2)$} (B);
\path (D) edge[bend left] node[below] {$(1,3)$} (B);
\path (D) edge node[right] {$(1,2)$} (F);
\path (B) edge[bend left] node[above] {$(4,1)$} (D);
\path (H) edge node[left] {$(2,1)$} (G);
\path (I) edge[bend left] node[below] {$(2,4)$} (H);
\path (I) edge[bend right] node[above] {$(4,1)$} (H);
\path (G) edge node[above right] {$(3,2)$} (I);
\end{tikzpicture}
\\
\hline
\end{tabular}
\caption{The two first steps in the algorithm of Theorem
  \ref{thm:wangtobar}. $(1)$
  corresponds to the representation of the tileset $(I)$. $(2)$
  deletes edges $(u,v)$ from $(1)$ if no path from $v$ to $u$ exists.
  $(3)$ deletes equivalences classes (connected components) from $(2)$
  which are cycles. The transform from $(2)$ to $(3$) does \textbf{not} preserve tileability.}
\label{fig:algoun}
\end{figure}
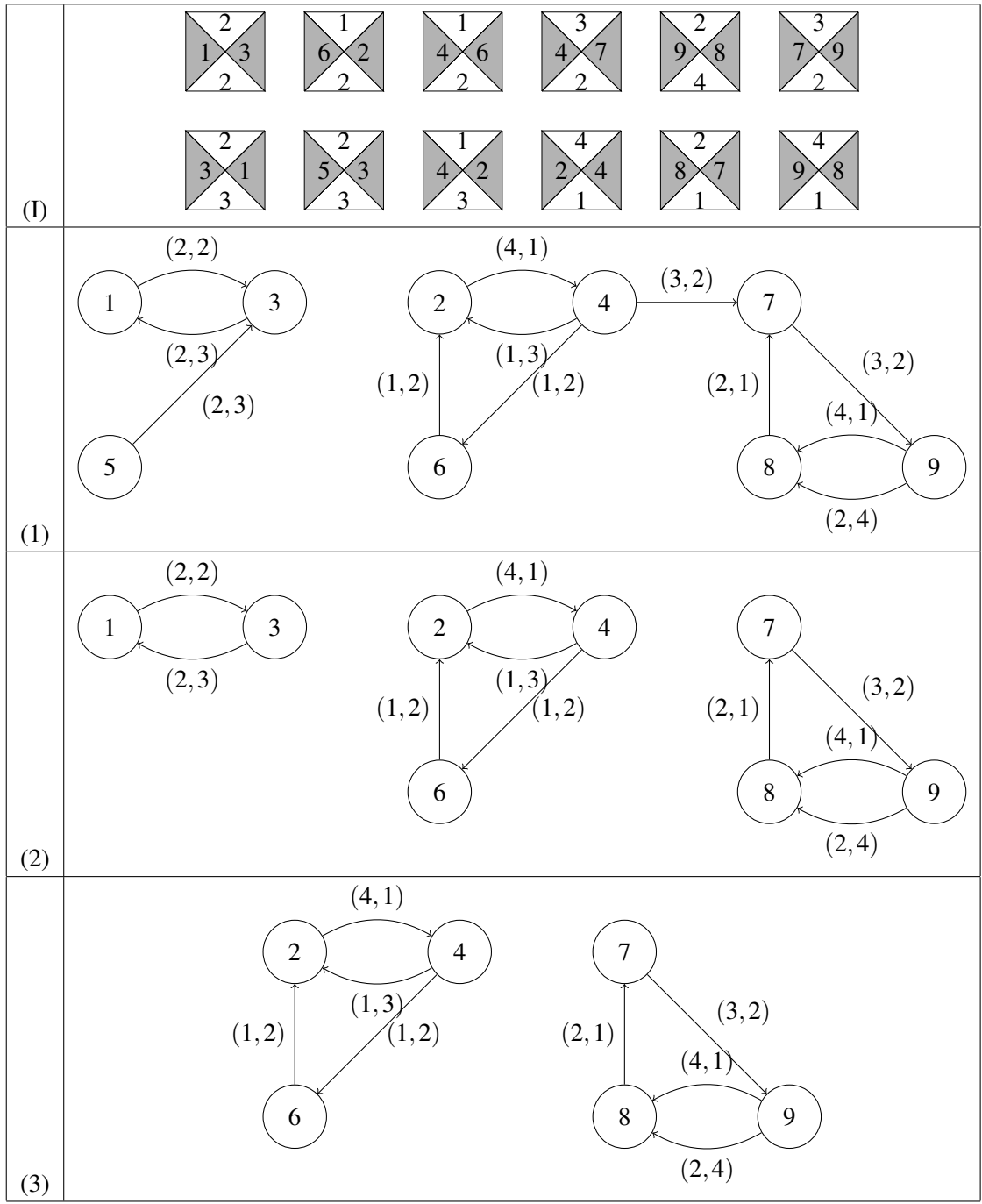

\begin{figure}[htbp]	
\begin{tabular}{|c|p{14cm}|}
\hline	
	(3)
& 
\begin{tikzpicture}[->, node distance=2.5cm]
\node[state] (B)  {2};
\node[state] (D) [right of=B] {4};
\node[state] (G) [right of=D] {7};
\node[state] (F) [below of=B] {6};
\node[state] (H) [below of=G] {8};
\node[state] (I) [right of=H] {9};
\path (F) edge node[left] {$(1,2)$} (B);
\path (D) edge[bend left] node[below] {$(1,3)$} (B);
\path (D) edge node[right] {$(1,2)$} (F);
\path (B) edge[bend left] node[above] {$(4,1)$} (D);
\path (H) edge node[left] {$(2,1)$} (G);
\path (I) edge[bend left] node[below] {$(2,4)$} (H);
\path (I) edge[bend right] node[above] {$(4,1)$} (H);
\path (G) edge node[above right] {$(3,2)$} (I);
\end{tikzpicture}
\\
\hline (4) &
\begin{tikzpicture}[->, node distance=2.5cm]
\node[state] (B)  {2};
\node[state] (D) [right of=B] {4};
\node[state] (G) [right of=D] {7};
\node[state] (I) [below right of=G] {9};
\path (D) edge[bend left] node[below] {$(1,3)$} (B);
\path (D) edge node[above] {$(11,22)$} (B);
\path (B) edge[bend left] node[above] {$(4,1)$} (D);
\path (I) edge[bend left] node[below left] {$(22,41)$} (G);
\path (I) edge[bend right] node[above right] {$(42,11)$} (G);
\path (G) edge node {$(3,2)$} (I);
\end{tikzpicture}
\\
\hline (5) & 
\begin{tikzpicture}[->, node distance=2.5cm]
\node[state] (D) [right of=B] {4};
\node[state] (I) [below right of=D] {9};
\path (D) edge[loop above] node[above] {$(114,221)$} (D);
\path (D) edge[loop below] node[below] {$(14,31)$} (D);
\path (I) edge[loop above] node[above] {$(223,412)$} (I);
\path (I) edge[loop below] node[below] {$(423,112)$} (I);
\end{tikzpicture}
\\
\hline (O) & 
\begin{tikzpicture}[scale=0.3]
\draw(0,0) rectangle (12,4);
\draw[line width=0.1pt, fill=black!30] (0,0) -- +(2,2) -- +(0,4) -- +(0,0);
\draw[line width=0.1pt, fill=black!30] (12,0) -- +(-2,2) -- +(0,4) -- +(0,0);
\draw (1,2) node {4};
\draw (11,2) node {4};
\draw (2,1) node {2};
\draw (6,1) node {2};
\draw (10,1) node {1};
\draw (2,3) node {1};
\draw (6,3) node {1};
\draw (10,3) node {4};
\begin{scope}[xshift=16cm]
\draw(0,0) rectangle (8,4);
\draw[line width=0.1pt, fill=black!30] (0,0) -- +(2,2) -- +(0,4) -- +(0,0);
\draw[line width=0.1pt, fill=black!30] (8,0) -- +(-2,2) -- +(0,4) -- +(0,0);
\draw (1,2) node {4};
\draw (7,2) node {4};
\draw (2,1) node {3};
\draw (6,1) node {1};
\draw (2,3) node {1};
\draw (6,3) node {4};
\end{scope}
\begin{scope}[yshift=-6cm]
\draw(0,0) rectangle (12,4);
\draw[line width=0.1pt, fill=black!30] (0,0) -- +(2,2) -- +(0,4) -- +(0,0);
\draw[line width=0.1pt, fill=black!30] (12,0) -- +(-2,2) -- +(0,4) -- +(0,0);
\draw (1,2) node {9};
\draw (11,2) node {9};
\draw (2,1) node {4};
\draw (6,1) node {1};
\draw (10,1) node {2};
\draw (2,3) node {2};
\draw (6,3) node {2};
\draw (10,3) node {3};
\end{scope}
\begin{scope}[yshift=-6cm,xshift=16cm]
\draw(0,0) rectangle (12,4);
\draw[line width=0.1pt, fill=black!30] (0,0) -- +(2,2) -- +(0,4) -- +(0,0);
\draw[line width=0.1pt, fill=black!30] (12,0) -- +(-2,2) -- +(0,4) -- +(0,0);
\draw (1,2) node {9};
\draw (11,2) node {9};
\draw (2,1) node {1};
\draw (6,1) node {1};
\draw (10,1) node {2};
\draw (2,3) node {4};
\draw (6,3) node {2};
\draw (10,3) node {3};
\end{scope}
\end{tikzpicture}
\\
\hline
\end{tabular}
\caption{The last steps in the algorithm of Theorem \ref{thm:wangtobar}. We
  basically contract vertices of outdegree one until there are no more
  such vertices. We contract vertices $6$ and $8$ from $(3)$ to $(4)$
  and $(2)$ and $(7)$ from $(4)$ to $(5)$. We then obtain the 4 Wang
  bars in $(O)$, from a initial parameter of $12-9 = 3$.}
\label{fig:algodeux}
\end{figure}
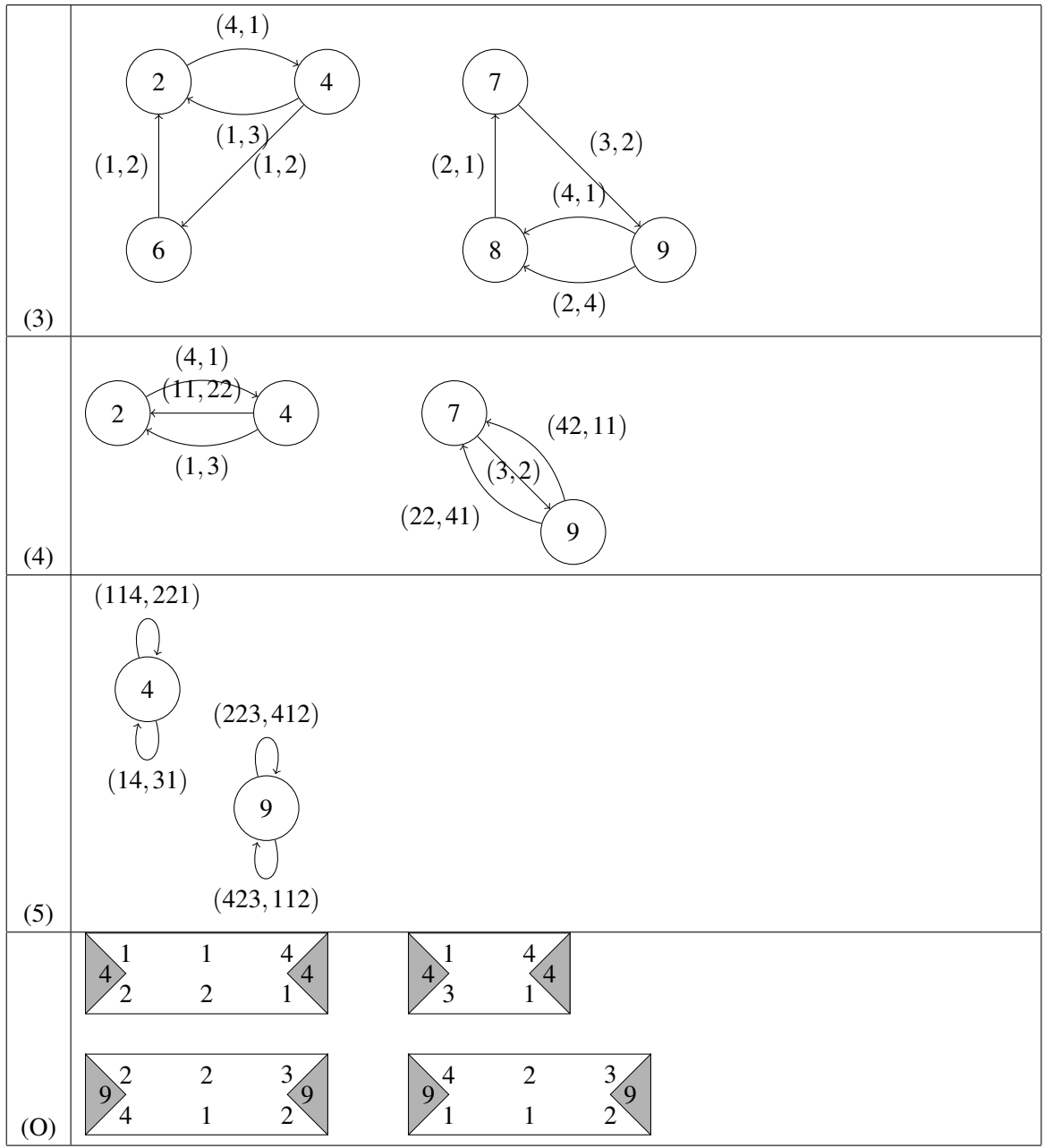
\begin{proof}
	We start from a tileset $W$ with parameter at most  $k$. If we
	rotate the tileset, we can suppose w.l.o.g that the maximum number
	of colors is reached on the west side.
	An example of the whole algorithm is provided on
	Fig.\ref{fig:algoun} and Fig.\ref{fig:algodeux}.

    Using Proposition\ref{prop:etapeun} and Proposition \ref{prop:etapedeux}, we obtain a
	tileset $V \subseteq W$ so that:
	\begin{itemize}
		\item Either $W$ tiles the plane periodically
		\item Or $V$ tiles the plane
		\item Or $W$ does not tile the plane
	\end{itemize}

By definition of the parameter, it is clear that the parameter of $V$
is less or equal to the parameter of $W$.
We now look at $V$ as a graph $G$.
Let $e(G)$ be the number of edges of $G$ and $v(G)$ its number of
vertices. By definition of $G$, $e(G) - v(G) \leq k$.
We now look at each equivalence class (strongly connected component) 
for $\sim_G$. 
We now execute the following algorithm on $G$, changing it into 
a $\omega$-labeled graph:

\begin{itemize}
\item If the number of edges in each equivalence class is at least
  twice the number of vertices, then we are done.
\item Otherwise let $C$ be such an equivalence class, and $e(C)$ (resp. 
$v(C)$) its number of edges (resp. vertices). We now have $e(C) < 2v(C)$.
\item This implies that $v(C) > 1$, otherwise all vertices of $C$ 
  would have outdegree $1$, and such
equivalences classes have been eliminated.
\item By a simple counting argument, there exists a vertex of outdegree
less than $2$ (the sum of all outdegrees
is equal to $e(C)$), that is of outdegree $1$.
\item Let $u$ be this vertex. The edge leaving $u$ cannot be from $u$ to $u$:
  There exists some other vertex $v$ in $C$, and there is a path from
  $u$ to $v$, which would be impossible if the only vertex leaving $u$ came back to $u$.
\clearpage
\item The edge from $u$ is then going to some $v \not= u$. Let $(x,y)$
  be its label. We now change the graph in the following way:
  \begin{itemize}
	  \item We delete the edge leaving from $u$.
	  \item For each edge from $w$ to $u$ labeled by $(u_1, u_2)$,
		we delete this edge and add a new edge from $w$ to $v$,
		labeled with $(u_1 x, u_2y)$. It is clear that this does not
		change tileability by the $\omega$-labeled graph.
	  \item The new graph has exactly one less vertex and one
		less edge than the previous one.
	  \item Doing this does not change the property of the graph: It
		is still true that a path from $u$ to $v$ implies a path from
		$v$ to $u$, and that there exist no equivalence class where
		every vertex is of outdegree one.
  \end{itemize}
  \item We repeat until the number of edges in each equivalence class
	is at least twice the number of vertices.
\end{itemize}
Now the $\omega$-labeled graph $G'$ we obtain has $n$ edges and $p$
vertices, with $n \geq 2p$. Furthermore, as we always delete a vertex
and a edge at the same time, $n-p \leq k$. We deduce that $n = 2n - 2p + 2p - n \leq 2k + 0 \leq 2k$.
So $G'$ represents a set of at most $2k$ bars. Let $B$ be this set.

Now the barset $B$ has the following property:
\begin{itemize}
	\item Either $W$ tiles periodically
	\item Or $B$ tiles the plane (in which case $W$ tiles the plane)
	\item Or $W$ does not tile the plane
\end{itemize}
So we first ask the oracle whether $B$ tiles the plane. If it does,
then $W$ tiles the plane. Otherwise we know that either $W$ tiles the
plane periodically, or it does not tile the plane. We can test for the
two concurrently until one of them halts.
\end{proof}

To recap, if we are able to decide if a tileset of parameter $k$ tiles
the plane, we can decide whether a barset of $k+1$ Wang bars tiles
the plane. Conversely, if we are able to decide if a barset of $k$
Wang bars tiles the plane, we can decide if a tileset of parameter
$k/2$ tiles the plane.
\clearpage
\section{Tileability for 44 Wang Bars is undecidable}

We will now prove that there is no algorithm that decides whether a
set of 44 Wang Bars tiles the plane.
This implies that there exists no algorithm that decide whether a
tileset of parameter at most 43 tiles the plane.

We will do a reduction from the Domino Problem. That is we will
explain how to encode any tileset $\tau$ into a barset $B_\tau$ of 44 Wang
bars (and thus into a tileset of parameter at most 43) so that $\tau$ tiles the plane iff $B$ tiles the plane.

The reduction we use here is heavily inspired by the transformation by
Ollinger \cite{OllingerPolyo} from any tileset to a set of 11 polyominoes. The main
difference is that we need here to make the polyominoes ``flat''.

Let $\tau$ be a set of $n$ Wang tiles. To define $B_\tau$ easier, all bars
of $B_\tau$ will have the property that
the color on the east and west side will be the same for all bars, that is
the only constraint to put two bars together are vertical constraints.
This means that we can see a bar $b$ as a pair of words of the same length.

Now we give the transformation. Let $n$ be the number of tiles in $\tau$ and
$C$ an upper bound on the number of colors, and suppose the colors
are numbered from $1$ to $C$.

The barset will first contain 16 bars, that do not depend on the
tileset $\tau$. These bars are depicted in Fig.\ref{fig:barun}. The colors
${a,b,c,d,e,f,g}$ that appear on these bars appear only in these bars,
so the only way to arrange the bars in a tiling is as depicted on the figure. We
will omit the labels ${a,b,c,d,e,f,g}$ on subsequent pictures.

\newcommand \barre{
\fill[fill=green!30,very thick] (0,0.5) rectangle ++(14,0.5);
\fill[fill=red!30,very thick] (0,0) rectangle ++(15,0.5);
\fill[fill=green!30,very thick] (17,0.5) rectangle ++(14,0.5);
\fill[fill=red!30,very thick] (16,0) rectangle ++(15,0.5);
\fill[fill=green,very thick] (0,4) rectangle ++(14,0.5);
\fill[fill=green,very thick] (17,4) rectangle ++(14,0.5);
\fill[very thick,fill=blue] (15,-3) rectangle ++(1,0.5);
\fill[very thick,fill=blue] (15,7.5) rectangle ++(1,0.5);
\draw[very thick] (0,0) -- ++(15,0) -- ++(0,-3) -- ++ (1,0) --
++ (0,3) -- ++ (15,0) -- ++ (0,1) -- ++ (-14,0)  -- ++(0,1) -- ++ (-1, 0) -- ++ (0,1) -- ++ (1,0) -- ++ (0,1) -- ++ (14,0)
-- ++ (0,1) -- ++ (-15,0) -- ++ (0,3) -- ++ (-1,0) -- ++ (0,-3) -- ++
(-15,0) -- ++ (0,-1) -- ++ (14,0) -- ++ (0,-1) -- ++ (1, 0) -- ++ (0,-1) -- ++ (-1,0) -- ++ (0,-1)  -- ++ (-14,0) -- cycle;
}
\newcommand \filler{
\fill[fill=green] (1,1.5) rectangle ++ (7, .5);
\fill[fill=green!30] (1,-1) rectangle ++ (7, .5);
\draw[very thick] (0,0) -- ++ (1,0) -- ++ (0,-1) -- ++ (7,0) -- ++ (0,1) -- ++ (-1,0) -- ++ (0,1) -- ++ (1,0) -- ++(0,1) -- ++ (-7, 0) -- ++ (0,-1) -- ++ (-1,0) -- cycle;
}
\newcommand \fillersym{
\fill[fill=green] (-1,1.5) rectangle ++ (-7, .5);
\fill[fill=green!30] (-1,-1) rectangle ++ (-7, .5);
\draw[very thick] (0,0) -- ++ (-1,0) -- ++ (0,-1) -- ++ (-7,0) -- ++ (0,1) -- ++ (1,0) -- ++ (0,1) -- ++ (-1,0) -- ++(0,1) -- ++ (7, 0) -- ++ (0,-1) -- ++ (1,0) -- cycle;
}
\newcommand \filla{
\fill[black!30] (0,0) rectangle ++ (1,.5);
\draw[very thick] (0,0) rectangle ++ (1, 2);
}

\newcommand \fillb{
\fill[fill=black!30] (0,0) rectangle ++ (1,.5);
\fill[fill=red!30] (0,2.5) rectangle ++ (1,.5);
\draw[very thick] (0,0) rectangle ++ (1, 3);}

\newcommand \fillc{
\fill[fill=red!30] (0,0) rectangle ++ (1,.5);
\fill[fill=black] (0,1.5) rectangle ++ (1,.5);
\draw[very thick] (0,0) rectangle ++ (1, 2);}

\newcommand \filld{
\fill[fill=black] (0,2.5) rectangle ++ (1,.5);
\draw[very thick] (0,0) rectangle ++ (1, 3);}
\newcommand\content{
      \draw[very thick] (0,0) rectangle ++(23,1);
      \draw[very thick,fill=black] (1,0) rectangle ++(21,.5);
      \draw[very thick,fill=black!30] (1,0.5) rectangle ++(21,.5);
\draw[fill=blue] (4,0) rectangle ++ (1,1);
\draw[fill=blue] (11,0) rectangle ++ (1,1);
\draw[fill=blue] (18,0) rectangle ++ (1,1);
\draw[fill=yellow] (1, 0.5) rectangle ++ (1,.5);
\draw[fill=yellow] (7, 0.5) rectangle ++ (1,.5);
\draw[fill=yellow] (6, 0) rectangle ++ (1,.5);
\draw[fill=yellow] (2, 0) rectangle ++ (1,.5);
\draw[fill=yellow] (8, 0.5) rectangle ++ (1,.5);
\draw[fill=yellow] (13, 0.5) rectangle ++ (1,.5);
\draw[fill=yellow] (13, 0) rectangle ++ (1,.5);
\draw[fill=yellow] (10, 0) rectangle ++ (1,.5);
\draw[fill=yellow] (16, 0.5) rectangle ++ (1,.5);
\draw[fill=yellow] (19, 0.5) rectangle ++ (1,.5);
\draw[fill=yellow] (19, 0) rectangle ++ (1,.5);
\draw[fill=yellow] (15, 0) rectangle ++ (1,.5);
}
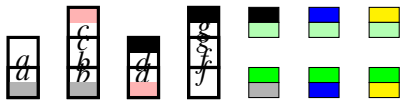
\begin{figure}[htbp]
\begin{center}
	\begin{tikzpicture}[scale=0.4]
\draw[very thick] (0,0) rectangle ++(1,1);
\draw  (.5,.8) node {$a$};
\draw  (.5,1.2) node {$a$};
\filla
\begin{scope}[xshift=2cm]
\draw[very thick] (0,0) rectangle ++(1,1);
\draw[very thick] (0,1) rectangle ++(1,1);
\draw  (.5,.8) node {$b$};
\draw  (.5,1.2) node {$b$};
\draw  (.5,1.8) node {$c$};
\draw  (.5,2.2) node {$c$};
\fillb
\end{scope}
\begin{scope}[xshift=4cm]
\draw[very thick] (0,0) rectangle ++(1,1);
\draw  (.5,.8) node {$d$};
\draw  (.5,1.2) node {$d$};
\fillc
\end{scope}
\begin{scope}[xshift=6cm]
\draw[very thick] (0,0) rectangle ++(1,1);
\draw[very thick] (0,1) rectangle ++(1,1);
\draw  (.5,.8) node {$f$};
\draw  (.5,1.2) node {$f$};
\draw  (.5,1.8) node {$g$};
\draw  (.5,2.2) node {$g$};
\filld
\end{scope}
	\begin{scope}[xshift= 7cm] 
\draw[fill=green] (1,0.5) rectangle ++ (1,.5);
\draw[fill=black!30] (1,0) rectangle ++ (1,.5);

\draw[fill=green] (3,0.5) rectangle ++ (1,.5);
\draw[fill=blue] (3,0) rectangle ++ (1,.5);

\draw[fill=green] (5,0.5) rectangle ++ (1,.5);
\draw[fill=yellow] (5,0) rectangle ++ (1,.5);

\draw[fill=black] (1,2.5) rectangle ++ (1,.5);
\draw[fill=green!30] (1,2) rectangle ++ (1,.5);

\draw[fill=blue] (3,2.5) rectangle ++ (1,.5);
\draw[fill=green!30] (3,2) rectangle ++ (1,.5);

\draw[fill=yellow] (5,2.5) rectangle ++ (1,.5);
\draw[fill=green!30] (5,2) rectangle ++ (1,.5);
\end{scope}
\end{tikzpicture}
\end{center}
\caption{The first $16$ bars}
\label{fig:barun}
\end{figure}
The next tiles are almost the same for any tileset, and they just
depend on the number of tiles $n$, and the number of colors $C$. They
are depicted in Fig. \ref{fig:bardeux}. As before, lowercase letters and greek
letters appear in only two tiles so that the only way to arrage the
bars is as depicted. We will not see them as $27$ bars, but as $5$
big blocks. The first one will be called the \emph{box}, the two next ones
\emph{fillers}, and the last ones \emph{handles}.
\newcommand\handle{
  \fill[fill=red!30] (1,2.5) rectangle ++(-15,.5);
  \fill[fill=yellow] (0,0) rectangle ++(1,.5);
  \fill[fill=yellow] (-15,4.5) rectangle ++(1,.5);
  \draw [very thick] (0,0) -- ++ (1,0) -- ++ (0,3) -- ++ (-15,0)
  -- ++ (0,2) -- ++ (-1,0) -- ++ (0,-3) -- ++ (15,0) -- ++ (0,-2);
  }

\newcommand\handlesym{
  \fill[fill=red!30] (-1,2.5) rectangle ++(15,.5);
  \fill[fill=yellow] (0,0) rectangle ++(-1,.5);
  \fill[fill=yellow] (15,4.5) rectangle ++(-1,.5);
  \draw [very thick] (0,0) -- ++ (-1,0) -- ++ (0,3) -- ++ (15,0)
  -- ++ (0,2) -- ++ (1,0) -- ++ (0,-3) -- ++ (-15,0) -- ++ (0,-2);
  }

\begin{figure}[htbp]
\begin{tikzpicture}[scale=0.4]
	\draw (20,-2) node {The \textbf{box}};
	\barre
\draw[very thick] (15,-3) rectangle ++ (1,1);
\draw[very thick] (15,-2) rectangle ++ (1,1);
\draw[very thick] (15,-1) rectangle ++ (1,1);
\draw[very thick] (14,1) rectangle ++ (3,1);
\draw[very thick] (14,3) rectangle ++ (3,1);
\draw[decoration={brace, raise=0.5cm},decorate] (0,4) -- ++ (14,0);
\draw (7,6) node {$(2C+1)(n-1) $};
\draw (15.5,-2.2) node {$h$};
\draw (15.5,-1.8) node {$h$};
\draw (15.5,-1.2) node {$i$};
\draw (15.5,-0.8) node {$i$};
\draw (15.5,-0.2) node {$j$};
\draw (15.5,0.2) node {$j$};
\draw (14.5,0.8) node {$k$};
\draw (14.5,1.2) node {$k$};
\draw (15.5,0.8) node {$k$};
\draw (15.5,1.2) node {$k$};
\draw (16.5,0.8) node {$k$};
\draw (16.5,1.2) node {$k$};
\draw (15.5,1.8) node {$l$};
\draw (15.5,2.2) node {$l$};
\draw (15.5,2.8) node {$m$};
\draw (15.5,3.2) node {$m$};
\draw (14.5,3.8) node {$n$};
\draw (14.5,4.2) node {$n$};
\draw (15.5,3.8) node {$n$};
\draw (15.5,4.2) node {$n$};
\draw (16.5,3.8) node {$n$};
\draw (16.5,4.2) node {$n$};
\draw[very thick] (15,5) rectangle ++ (1,1);
\draw[very thick] (15,6) rectangle ++ (1,1);
\draw (15.5,4.8) node {$o$};
\draw (15.5,5.2) node {$o$};
\draw (15.5,5.8) node {$p$};
\draw (15.5,6.2) node {$p$};
\draw (15.5,6.8) node {$q$};
\draw (15.5,7.2) node {$q$};

\begin{scope}[yshift=-10cm]
	\filler
	\draw (24,0) node {The \textbf{fillers}};
\draw[very thick] (0,0) rectangle ++ (7,1);
\draw[decoration={brace, raise=0.5cm},decorate] (1,1) -- ++ (7,0);
\draw (5,3) node {$2C+1$};
\draw (1.5,-0.2) node {$r$};
\draw (2.5,-0.2) node {$r$};
\draw (3.5,-0.2) node {$r$};
\draw (4.5,-0.2) node {$r$};
\draw (5.5,-0.2) node {$r$};
\draw (6.5,-0.2) node {$r$};
\draw (1.5,0.2) node {$r$};
\draw (2.5,0.2) node {$r$};
\draw (3.5,0.2) node {$r$};
\draw (4.5,0.2) node {$r$};
\draw (5.5,0.2) node {$r$};
\draw (6.5,0.2) node {$r$};
\draw (1.5,0.8) node {$s$};
\draw (2.5,0.8) node {$s$};
\draw (3.5,0.8) node {$s$};
\draw (4.5,0.8) node {$s$};
\draw (5.5,0.8) node {$s$};
\draw (6.5,0.8) node {$s$};
\draw (1.5,1.2) node {$s$};
\draw (2.5,1.2) node {$s$};
\draw (3.5,1.2) node {$s$};
\draw (4.5,1.2) node {$s$};
\draw (5.5,1.2) node {$s$};
\draw (6.5,1.2) node {$s$};
\end{scope}
\begin{scope}[yshift=-10cm,xshift=20cm]
	\fillersym
\draw[very thick] (-7,0) rectangle ++ (7,1);
\draw[decoration={brace, raise=0.5cm},decorate] (-8,1) -- ++ (7,0);
\draw (-5,3) node {$2C+1$};
\draw (-1.5,-0.2) node {$t$};
\draw (-2.5,-0.2) node {$t$};
\draw (-3.5,-0.2) node {$t$};
\draw (-4.5,-0.2) node {$t$};
\draw (-5.5,-0.2) node {$t$};
\draw (-6.5,-0.2) node {$t$};
\draw (-1.5,0.2) node {$t$};
\draw (-2.5,0.2) node {$t$};
\draw (-3.5,0.2) node {$t$};
\draw (-4.5,0.2) node {$t$};
\draw (-5.5,0.2) node {$t$};
\draw (-6.5,0.2) node {$t$};
\draw (-1.5,0.8) node {$u$};
\draw (-2.5,0.8) node {$u$};
\draw (-3.5,0.8) node {$u$};
\draw (-4.5,0.8) node {$u$};
\draw (-5.5,0.8) node {$u$};
\draw (-6.5,0.8) node {$u$};
\draw (-1.5,1.2) node {$u$};
\draw (-2.5,1.2) node {$u$};
\draw (-3.5,1.2) node {$u$};
\draw (-4.5,1.2) node {$u$};
\draw (-5.5,1.2) node {$u$};
\draw (-6.5,1.2) node {$u$};
\end{scope}
\begin{scope}[yshift=-20cm, xshift=15cm]
	\draw[very thick] (0,0) rectangle ++ (1,1);
	\draw[very thick] (0,1) rectangle ++ (1,1);
	\draw (0.5,0.8) node {$v$};
	\draw (0.5,1.2) node {$v$};
	\draw (0.5,1.8) node {$w$};
	\draw (0.5,2.2) node {$w$};
	\draw[very thick] (-15,3) rectangle ++ (1,1);
	\draw (-14.5,2.8) node {$x$};
	\draw (-14.5,3.2) node {$x$};
	\draw (-14.5,3.8) node {$y$};
	\draw (-14.5,4.2) node {$y$};
\draw[decoration={brace, raise=0.5cm},decorate] (-14,2) -- ++ (14,0);
\draw (-7,4) node {$(2C+1)(n-1)$};
	\handle
	\draw (11,0) node {The \textbf{handles}};	
\end{scope}	

\begin{scope}[yshift=-24cm, xshift=6cm]
	\draw[decoration={brace, mirror, raise=0.5cm},decorate] (0,3) -- ++ (14,0);
	\draw (7,1) node {$(2C+1)(n-1)$};
	\draw[very thick] (-1,1) rectangle ++ (1,1);
	\draw (-0.5,0.8) node {$z$};
	\draw (-0.5,1.2) node {$z$};
	\draw (-0.5,1.8) node {$\alpha$};
	\draw (-0.5,2.2) node {$\alpha$};
	\draw[very thick] (14,3) rectangle ++ (1,1);
	\draw (14.5,2.8) node {$\beta$};
	\draw (14.5,3.2) node {$\beta$};
	\draw (14.5,3.8) node {$\gamma$};
	\draw (14.5,4.2) node {$\gamma$};
	\handlesym
\end{scope}	
\end{tikzpicture}	
\caption{The 27 next bars.}
\label{fig:bardeux}
\end{figure}

And now, to finish, the last bar, which depends on the exact tiles we
use.
For each tile $t$ of $\tau$ with north,west, east, south color respectively
$(a,b,c,d)$, we will first consider the following pair of
words ${(n_t, s_t) = (0^{a-1}Y0^{C-a}B0^{b-1}Y0^{C-b},
2^{c-1}Y2^{C-c}B2^{d-1}Y2^{C-d})}$
over the alphabet $\{0,Y,B\}$ of length $2C+1$. Basically, we encode
the colors in \emph{unary}.

For a given tileset $\tau$ we now consider the bar obtained by concatening all
pair of words $(n_t,s_t)$ for tiles $t \in \tau$, adding to them two
blank symbols to each extremity. Fig.~\ref{fig:content} gives an example. This
last bar will be called the \emph{content}.

\begin{figure}[hbtp]
\begin{center}

\end{center}
\caption{How the last bar is built}
\label{fig:content}
\end{figure}

Now that we have defined all the bars, we can explain how the
simulation works.
The idea is all these bars combine to form $n$ polyominoes, represented in
Fig. \ref{fig:poly}, that acts like the $n$ original Wang tiles.

\begin{figure}[htbp]
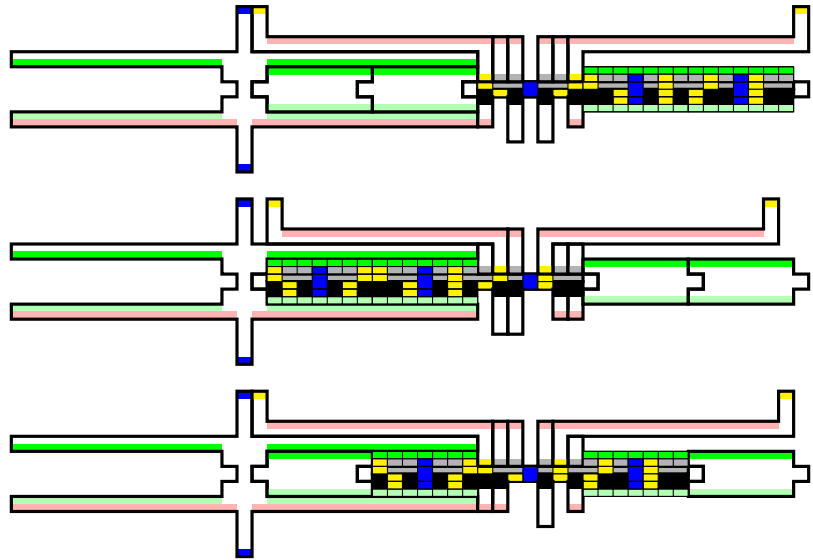

\begin{center}
%
	
\end{center}
\caption{How the bars can represent the original Wang tiles}
\label{fig:poly}
\end{figure}

The north and east side are encoded in unary above the polyomino, and the west
and south side below, so that from an original tiling of the plane by
the Wang tiles, we can obtain a tiling by the polyominoes, and thus by
the bars. An example is depicted in Fig. \ref{fig:bigscore}. It should now be clear that
if there is a tiling by the original tileset, then there is a tiling
by the set of Wang bars.
\begin{figure}
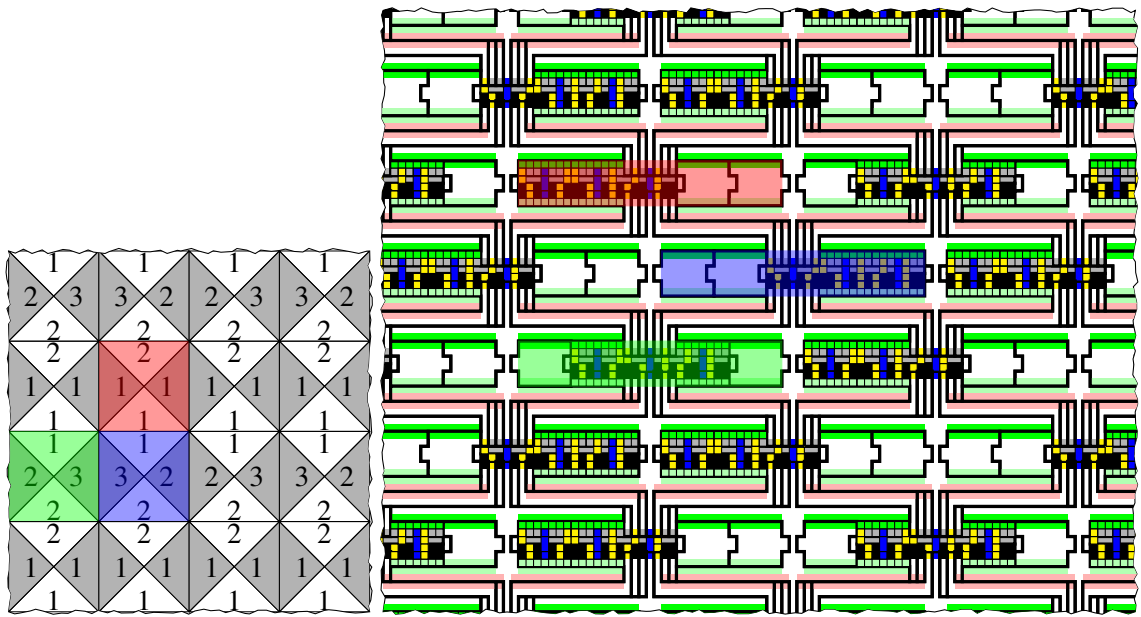

\begin{center}
%
\end{center}
\caption{How a tiling by the Wang tiles becomes a tiling by the bars}
\label{fig:bigscore}
\end{figure}

We will now explain why a tiling by the bars has to be of this form.
We will prove only that it must look as depicted in Fig.\ref{fig:poly}, the
rest of the proof being easy.

The different steps in the proof are depicated in Fig.\ref{fig:proof}.
We will first explain why the bar labeled $A$ must be there, then the
bar labeled $B$, etc.

First we look at all the bars, et we prove that we need to use the box.
It is quite clear:
\begin{itemize}
	\item The fillers must be attached to the box
	\item The first four polyominoes from the first set must be
	  attached to the content, and the last ones to the box
	\item The handles, due to the yellow color, must be attached either
	  to the content, or to a (green,yellow) bar that can only be
	  attached to the box.
	\item Finally, the content, due to the blue color, cannot be
	  attached to itself, so it must either be attached to the small
	  (blue, light green) bar (that in turn must be attached to the
	  box), or to the box directly.	
\end{itemize}	
In all cases, we have seen that any bar is linked, sooner or later,
to the box.

So we can now suppose that a box appear somewhere ($A$ on
Fig.\ref{fig:proof})

We know look at how to fill its upper right half. If we do not use a
handle, then the only possibility is to use the fourth tile from
Fig.\ref{fig:barun}. But we will obtain a row of at least
$(2C+1)(n-1)$ black colors, and the only bar with a black color on its
south side (the content) has at most $2C$ consecutive black colors. So
there must be a handle ($B$ on Fig.\ref{fig:proof}).

Now we look at how to fill the right part of the box. It is clear it can be
only filled with fillers or with the content (attached to small tiles).
Now we can use at most $(n-1)$ fillers: If we use one more, this
filler will be outside of the box, which is not possible as it will
collide with the handle.

So there must be somewhere a content ($C$ and $D$ on the figure). Due to the handle, the content
cannot be entirely inside boxes. So some of it is outside, and in
particular there is a blue color on the north side outside. The only
thing we can attach to this blue color is a box ($E$ on the figure). In particular, there
can be only \textbf{one} blue color appearing outside a box, which
means the only possible situation is the one depicted in Fig.
\ref{fig:proof}.

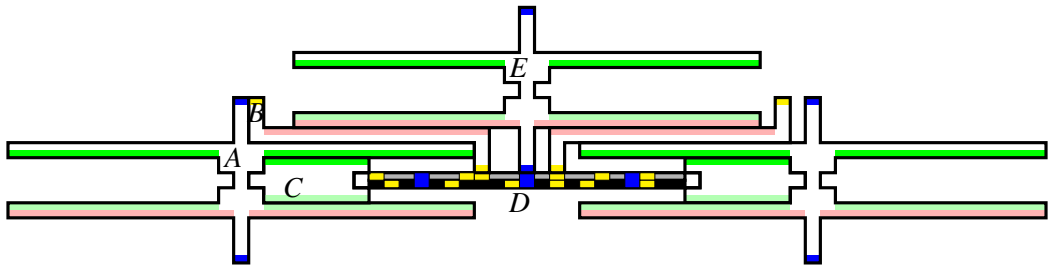
\begin{figure}[htbp]
\begin{center}
\begin{tikzpicture}[scale=0.2]
	\begin{scope}[yshift=-42cm,xshift= -19cm] \barre \end{scope}
\begin{scope}[yshift=-40cm,xshift= 4cm] \content \end{scope}
\begin{scope}[yshift=-40cm,xshift= -3cm] \filler \end{scope}
\begin{scope}[yshift=-40cm,xshift= 34cm] \fillersym \end{scope}
\begin{scope}[yshift=-39cm,xshift= 12cm] \handle \end{scope}
\begin{scope}[yshift=-36cm,xshift= 0cm] \barre \end{scope}
\begin{scope}[yshift=-39cm,xshift= 18cm] \handlesym \end{scope}
	\begin{scope}[yshift=-42cm,xshift= 19cm] \barre \end{scope}
\draw (-4,-38) node {$A$};
\draw (-2.5,-35) node {$B$};
\draw (0,-40) node {$C$};
\draw (15,-41) node {$D$};
\draw (15,-32) node {$E$};
\end{tikzpicture}
\end{center}
\caption{How to prove the bars must behave correctly.}
\label{fig:proof}
\end{figure}
It is then clear that this situation must force the bars (and
therefore the tiles they simulate) to be aligned correctly, so that the proof is done. We can now state the theorem and its corollary.

\begin{theorem}
	The tileability problem is undecidable for 44 bars. 
\end{theorem}	
\begin{corollary}
	The domino problem is undecidable for tilesets of parameter 43.
\end{corollary}

It is quite clear that the parameter 43 is not optimal.
There are indeed a few way to gain bars in the construction. For
instance, the two bars with upper side $v$ and $z$ can be
identified without changing anything. However we see no need in this
article to insert any reference to The Hitchhiker's Guide to the
Galaxy, and will stay with a parameter of 43.



\section{Tileability with 2 Wang bars is decidable}

In this last section, we will prove that we can decide whether a set
of two Wang bars tiles the plane. More precisely
\begin{theorem}
	A set of two Wang bars $B$ tiles the plane iff there is a periodic
	tiling by $B$.
\end{theorem}	

The rest of the section is devoted to the proof.

We start from two Wang bars $\alpha$ and $\beta$, et we will assume
that there exists a tiling with $\alpha$ and $\beta$, but no periodic
tiling, and reach a contradiction.

First, we see that the colors on the east and west side of the bars
have to be the same: Otherwise, either we cannot put $\alpha$ and
$\beta$ on the same line (so that every tiling must be periodic
horizontally of period $lcm(|\alpha|, |\beta|)$, hence there is a periodic tiling), or we must always put
$\alpha$ after $\beta$ and $\beta$ after $\alpha$ on the same line 
(in which case every tiling is periodic horizontally of period $|\alpha|+|\beta|$, hence there is a periodic tiling).
So, as in the previous section, the only constraints to put the two
bars together are vertical constraints, and we can see the bars as two
words $\alpha = (\alpha_N, \alpha_S)$, $\beta = (\beta_N, \beta_S)$.

W.l.o.g., we will assume $|\alpha| \geq |\beta|$.
We can also suppose that we cannot tile using only $\alpha$ (resp.
$\beta$), as this implies that there exists a periodic tiling.

\begin{lemma}
	\label{lemma:gamma}
	If the following pattern appears in some tiling, then there
	exists a periodic tiling
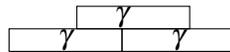
\begin{center}	
\begin{tikzpicture}[scale=0.3]
	\draw (0,0) rectangle ++(5,1);
	\draw (2.5,0.5) node {$\gamma$};
	\draw (5,1.5) node {$\gamma$};
	\draw (7.5,0.5) node {$\gamma$};
	\draw (5,0) rectangle ++(5,1);
	\draw (3,1) rectangle ++(5,1);
\end{tikzpicture}
\end{center}
where $\gamma$ is one of the two bars $\alpha$ or $\beta$.
\end{lemma}	
\begin{proof}
It is quite clear from the picture.
Let $n$ be the length of the bars, and $p$ the horizontal distance
between the lower left bar and the upper bar.

Looking at the upper bar, we write $\gamma_S = uv$ where $u$ is of size $n-p$
and $v$ of size $p$. Looking at the two other bar, we see that the
last $p$ symbols of $\gamma_N$ are $u$, and the first $n-p$ symbols of
$\gamma_N$ are $v$, so that $\gamma = (vu, uv)$. So a row containing only $\gamma$ will have the same upper and lower
side, up to a shift of $p$, and hence we can obtain a periodic tiling by
repeating this row with a suitable shift every time.
\end{proof}	

\begin{lemma}
	If there is a tiling where the following pattern (and its
	horizontal symmetry) does \textbf{not} appear,
	then there exists a periodic tiling.
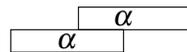
\begin{center}	
\begin{tikzpicture}[scale=0.3]
	\draw (0,0) rectangle ++(5,1);
	\draw (2.5,0.5) node {$\alpha$};
	\draw (5,1.5) node {$\alpha$};
	\draw (3,1) rectangle ++(5,1);
\end{tikzpicture}
\end{center}		
\end{lemma}	
\begin{proof}
We consider this tiling, and $p \geq 1$ the minimal number of consecutive
bars of type $\alpha$ in such a tiling. (it is clear that $p$ cannot be
infinite. Otherwise this imply that a bar of type $\alpha$ can be followed
(or preceded) at most once by a bar of type $\beta$ on each row, hence by the
same argument as Lemma \ref{lemma:twice}, there exists a tiling where
this doesn't appear, i.e. where every row consists either entirely of
bars of type $\alpha$ or entirely of bars of type $\beta$, which would
imply the existence of a periodic tiling).

Now we look at a position where this minimum number appears.
We take $p = 2$ for the picture.
\begin{center}	
\begin{tikzpicture}[scale=0.3]
	\draw (0,0) rectangle ++(5,1);
	\draw (2.5,0.5) node {$\alpha$};
	\draw (-5,0) rectangle ++(5,1);
	\draw (-2.5,0.5) node {$\alpha$};
\end{tikzpicture}
\end{center}	
By definition of $p$, these tiles are surrounded horizontally by bars
of type $\beta$.
\begin{center}	
\begin{tikzpicture}[scale=0.3]
	\draw (0,0) rectangle ++(5,1);
	\draw (2.5,0.5) node {$\alpha$};
	\draw (-5,0) rectangle ++(5,1);
	\draw (-2.5,0.5) node {$\alpha$};
	\draw (-8,0) rectangle ++(3,1);
	\draw (-6.5,0.5) node {$\beta$};
	\draw (5,0) rectangle ++(3,1);
	\draw (6.5,0.5) node {$\beta$};
\end{tikzpicture}
\end{center}	
Now by the hypothesis of the lemma, all bars above the bars of type $\alpha$
must be $\beta$ tiles:
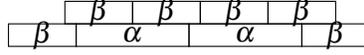
\begin{center}	
\begin{tikzpicture}[scale=0.3]
	\draw (0,0) rectangle ++(5,1);
	\draw (2.5,0.5) node {$\alpha$};
	\draw (-5,0) rectangle ++(5,1);
	\draw (-2.5,0.5) node {$\alpha$};
	\draw (-8,0) rectangle ++(3,1);
	\draw (-6.5,0.5) node {$\beta$};
	\draw (5,0) rectangle ++(3,1);
	\draw (6.5,0.5) node {$\beta$};
	\draw (-5.5,1) rectangle ++(3,1);
	\draw (-4,1.5) node {$\beta$};	
	\draw (-2.5,1) rectangle ++(3,1);
	\draw (-1,1.5) node {$\beta$};	
	\draw (.5,1) rectangle ++(3,1);
	\draw (2,1.5) node {$\beta$};		
	\draw (3.5,1) rectangle ++(3,1);
	\draw (5,1.5) node {$\beta$};			
\end{tikzpicture}
\end{center}	
Now we look at the bar we have to put above the lower rightmost
bar of type $\beta$. If it is of type  $\beta$, we obtain the pattern from the
previous lemma and we are done. Otherwise it is of type  $\alpha$, and by
the minimality of $p$, there are at least $p$ such bars:
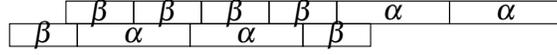
\begin{center}	
\begin{tikzpicture}[scale=0.3]
	\draw (0,0) rectangle ++(5,1);
	\draw (2.5,0.5) node {$\alpha$};
	\draw (-5,0) rectangle ++(5,1);
	\draw (-2.5,0.5) node {$\alpha$};
	\draw (-8,0) rectangle ++(3,1);
	\draw (-6.5,0.5) node {$\beta$};
	\draw (5,0) rectangle ++(3,1);
	\draw (6.5,0.5) node {$\beta$};
	\draw (-5.5,1) rectangle ++(3,1);
	\draw (-4,1.5) node {$\beta$};	
	\draw (-2.5,1) rectangle ++(3,1);
	\draw (-1,1.5) node {$\beta$};	
	\draw (.5,1) rectangle ++(3,1);
	\draw (2,1.5) node {$\beta$};		
	\draw (3.5,1) rectangle ++(3,1);
	\draw (5,1.5) node {$\beta$};			
	\draw (6.5,1) rectangle ++(5,1);
	\draw (9,1.5) node {$\alpha$};			
	\draw (11.5,1) rectangle ++(5,1);
	\draw (14,1.5) node {$\alpha$};			
\end{tikzpicture}
\end{center}	
To finish, we have again by the hypothesis of the lemma, that the bars
below it are of type $\beta$:
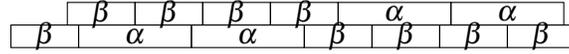
\begin{center}	
\begin{tikzpicture}[scale=0.3]
	\draw (0,0) rectangle ++(5,1);
	\draw (2.5,0.5) node {$\alpha$};
	\draw (-5,0) rectangle ++(5,1);
	\draw (-2.5,0.5) node {$\alpha$};
	\draw (-8,0) rectangle ++(3,1);
	\draw (-6.5,0.5) node {$\beta$};
	\draw (5,0) rectangle ++(3,1);
	\draw (6.5,0.5) node {$\beta$};
	\draw (-5.5,1) rectangle ++(3,1);
	\draw (-4,1.5) node {$\beta$};	
	\draw (-2.5,1) rectangle ++(3,1);
	\draw (-1,1.5) node {$\beta$};	
	\draw (.5,1) rectangle ++(3,1);
	\draw (2,1.5) node {$\beta$};		
	\draw (3.5,1) rectangle ++(3,1);
	\draw (5,1.5) node {$\beta$};			
	\draw (6.5,1) rectangle ++(5,1);
	\draw (9,1.5) node {$\alpha$};			
	\draw (11.5,1) rectangle ++(5,1);
	\draw (14,1.5) node {$\alpha$};
	\draw (8,0) rectangle ++(3,1);
	\draw (9.5,0.5) node {$\beta$};				
	\draw (11,0) rectangle ++(3,1);
	\draw (12.5,0.5) node {$\beta$};				
	\draw (14,0) rectangle ++(3,1);
	\draw (15.5,0.5) node {$\beta$};				
\end{tikzpicture}
\end{center}

Now we look at the picture without the bottom leftmost bar of type
$\beta$. We
obtain a pattern of $k$ bars of type $\beta$ followed by
$p$ bars of type $\alpha$ on
the first row, and of $p$ bars of type $\alpha$ followed by $k$ bars
of type $\beta$
on the second row (it is easy to see it must be the exact same $k$).
Now this pattern can be repeated periodically: The last few symbols of the
north part of the bar $\beta$  are indeed the same as the first symbols of
the south part of the bar $\beta$, as is witnessed by the bottom
leftmost bar of type $\beta$ we just deleted. Hence this pattern of $k$
bars of type $\beta$ followed by $p$ bars of type $\alpha$ can be repeated
periodically horizontally and vertically up to shift, and we obtain a
periodic tiling.
\end{proof}	

So we are done to our last case: The following pattern appears
somewhere:
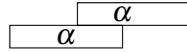
\begin{center}	
\begin{tikzpicture}[scale=0.3]
	\draw (0,0) rectangle ++(5,1);
	\draw (2.5,0.5) node {$\alpha$};
	\draw (5,1.5) node {$\alpha$};
	\draw (3,1) rectangle ++(5,1);
\end{tikzpicture}
\end{center}		
We look at some occurence of this pattern.
The lower left bar of type $\alpha$ is perhaps preceded by some other bars of
type $\alpha $ so that there are $p_1$ of them in total, and the same goes for
the upper right bar, which might be constituted of $p_2$¬†bars of type  $\alpha$.
Upto a rotation of the bars, we suppose that $p=\min(p_1,p_2) = p_1$. We will now not
care at all about the remaining $p_2 - p$ bars of type $\alpha$ that
might appear on the upper right.
The figure is now as follows (we take $p = 2$):
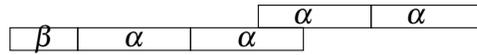
\begin{center}	
\begin{tikzpicture}[scale=0.3]
	\draw (0,0) rectangle ++(5,1);
	\draw (-5,0) rectangle ++(5,1);
	\draw (-8,0) rectangle ++(3,1);
	\draw (2.5,0.5) node {$\alpha$};
	\draw (5,1.5) node {$\alpha$};
	\draw (-2.5,0.5) node {$\alpha$};
	\draw (-6.5,0.5) node {$\beta$};
	\draw (10,1.5) node {$\alpha$};
	\draw (3,1) rectangle ++(5,1);
	\draw (8,1) rectangle ++(5,1);
\end{tikzpicture}
\end{center}		
(The leftmost bar of type $\alpha$ must of course be preceded by a bar
of type $\beta$ by the definition of $p$)

Now to complete each bar of type $\alpha$, we need to put at least one
bar of type $\beta$, otherwise we would obtain the figure from Lemma \ref{lemma:gamma}.
We thus put $q_1$ bars of type $\beta$ at the top left of the picture, and $q_2$ bars
at the bottom right. 
There are two cases:
\begin{itemize}
	\item There are too many bars of type $\beta$ on both side, so that
	  they completely cover the $p$ bars of type $\alpha$.
	  Let $q < \min(q_1, q_2)$  be such that the first $q-1$ bars of type
	  $\beta$ will not cover all bars of type $\alpha$ but $q$ such
	  bars will. We now have the following picture (with $q = 3$)
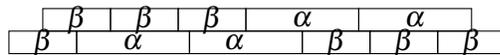
\begin{center}	
\begin{tikzpicture}[scale=0.3]
	\draw (0,0) rectangle ++(5,1);
	\draw (-5,0) rectangle ++(5,1);
	\draw (-8,0) rectangle ++(3,1);
	\draw (-6.5,0.5) node {$\beta$};
	\draw (2,0.5) node {$\alpha$};	
	\draw (5,1.5) node {$\alpha$};
	\draw (-2.5,0.5) node {$\alpha$};
	\draw (10,1.5) node {$\alpha$};
	\draw (2.5,1) rectangle ++(5,1);
	\draw (7.5,1) rectangle ++(5,1);
	\draw (-.5,1) rectangle ++(3,1);	
	\draw (-3.5,1) rectangle ++(3,1);	
	\draw (-6.5,1) rectangle ++(3,1);	
	\draw (5,0) rectangle ++(3,1);	
	\draw (8,0) rectangle ++(3,1);	
	\draw (11,0) rectangle ++(3,1);	
	\draw (1,1.5) node {$\beta$};
	\draw (-2,1.5) node {$\beta$};
	\draw (-5,1.5) node {$\beta$};
	\draw (6.5,0.5) node {$\beta$};
	\draw (9.5,0.5) node {$\beta$};
	\draw (12.5,0.5) node {$\beta$};
\end{tikzpicture}
\end{center}		
But by the same reasoning as the previous lemma, this implies that we
can tile using $q$ bars of type $\beta$ and $p$ bars of type $\alpha$
periodically
\item There is at least one side on which there are not enough bars of
  type $\beta$ to cover the bars of type $\alpha$. We take again $q =
  \min(q_1,q_2)$ and suppose that $q = q_1$ (the other case is
  similar, as our argument will not use the bottom left bar of type
  $\beta$), and we completely forget about the $q_2 - q_1$ remaining bars of type
  $\beta$.
  We take $p'\leq p$ so that the $p'-1$th bar is
  entirely covered by the bars of type $\beta$
  but not the $p'$th, to obtain the following picture ($p' = 1$ on the picture)
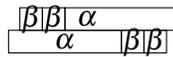
\begin{center}
\begin{tikzpicture}[scale=0.3]
	\draw (0,0) rectangle ++(5,1);
	\draw (2.5,0.5) node {$\alpha$};
	\draw (3.5,1.5) node {$\alpha$};
	\draw (2.5,1) rectangle ++(5,1);
	\draw (.5,1) rectangle ++ (1,1);
	\draw (1,1.5) node {$\beta$};
	\draw (1.5,1) rectangle ++ (1,1);
	\draw (2,1.5) node {$\beta$};	
	\draw (5,0) rectangle ++ (1,1);
	\draw (5.5,0.5) node {$\beta$};
	\draw (6,0) rectangle ++ (1,1);
	\draw (6.5,0.5) node {$\beta$};	
\end{tikzpicture}
\end{center}		

Now by definition of $q$, there must be a $\alpha$-bar on the top left:
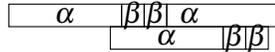
\begin{center}	
\begin{tikzpicture}[scale=0.3]
	\draw (0,0) rectangle ++(5,1);
	\draw (2.5,0.5) node {$\alpha$};
	\draw (3.5,1.5) node {$\alpha$};
	\draw (2.5,1) rectangle ++(5,1);
	\draw (.5,1) rectangle ++ (1,1);
	\draw (1,1.5) node {$\beta$};
	\draw (-4.5,1) rectangle ++ (5,1);
	\draw (-2,1.5) node {$\alpha$};
	\draw (1.5,1) rectangle ++ (1,1);
	\draw (2,1.5) node {$\beta$};	
	\draw (5,0) rectangle ++ (1,1);
	\draw (5.5,0.5) node {$\beta$};
	\draw (6,0) rectangle ++ (1,1);
	\draw (6.5,0.5) node {$\beta$};	
\end{tikzpicture}
\end{center}		
By the same argument as the previous lemma, this implies that we can
tile using $q$ bars $\beta$ followed by $p'$ bars $\alpha$
periodically in each row.
\end{itemize}
The proof is now finished, and in all cases, we proved that there
exists a periodic tiling.

\section{Conclusion} There is no algorithm to decide whether a set of
$44 $ Wang Bars tiles the plane. On the other hand, a set of $2$ Wang
Bars tiles the plane if and only if it tiles it periodically. It would
be interesting to reduce the gaps between these two values. It is
clear that $44$ is not optimal, and that we can reduce it. While we did
not work out the details, it is safe to assume we can obtain something
around 35 using the same idea. However, it is much more interesting to
try to prove that the problem is decidable for $3$ or $4$ Wang bars.
While it is our opinion that the result holds, this situation became
dramatically more complex with even $3$ Wang Bars, and we do not know
of any approach to solve this problem.

\paragraph{Acknowledgements} The first author thanks Daniel Gon\c{c}alves
and Pascal Vanier for some interesting discussions that lead to the
proof of the case of 2 Wang Bars.


\begin{thebibliography}{1}
\providecommand{\bibitemdeclare}[2]{}
\providecommand{\surnamestart}{}
\providecommand{\surnameend}{}
\providecommand{\urlprefix}{Available at }
\providecommand{\url}[1]{\texttt{#1}}
\providecommand{\href}[2]{\texttt{#2}}
\providecommand{\urlalt}[2]{\href{#1}{#2}}
\providecommand{\doi}[1]{doi:\urlalt{http://dx.doi.org/#1}{#1}}
\providecommand{\bibinfo}[2]{#2}

\bibitemdeclare{phdthesis}{BergerPhd}
\bibitem{BergerPhd}
\bibinfo{author}{Robert \surnamestart Berger\surnameend}
  (\bibinfo{year}{1964}): \emph{\bibinfo{title}{{The Undecidability of the
  Domino Problem}}}.
\newblock Ph.D. thesis, \bibinfo{school}{Harvard University}.

\bibitemdeclare{article}{Du}
\bibitem{Du}
\bibinfo{author}{Bruno \surnamestart Durand\surnameend} (\bibinfo{year}{1999}):
  \emph{\bibinfo{title}{{Tilings and Quasiperiodicity}}}.
\newblock {\sl \bibinfo{journal}{Theoretical Computer Science}}
  \bibinfo{volume}{221}(\bibinfo{number}{1-2}), pp. \bibinfo{pages}{61--75},
  \doi{10.1016/S0304-3975(99)00027-4}.

\bibitemdeclare{article}{Culik}
\bibitem{Culik}
\bibinfo{author}{Karel~Culik \surnamestart II\surnameend}
  (\bibinfo{year}{1996}): \emph{\bibinfo{title}{{An aperiodic set of 13 Wang
  tiles}}}.
\newblock {\sl \bibinfo{journal}{Discrete Mathematics}} \bibinfo{volume}{160},
  pp. \bibinfo{pages}{245--251}, \doi{10.1016/S0012-365X(96)00118-5}.

\bibitemdeclare{article}{Kari14}
\bibitem{Kari14}
\bibinfo{author}{Jarkko \surnamestart Kari\surnameend} (\bibinfo{year}{1996}):
  \emph{\bibinfo{title}{{A small aperiodic set of Wang tiles}}}.
\newblock {\sl \bibinfo{journal}{Discrete Mathematics}} \bibinfo{volume}{160},
  pp. \bibinfo{pages}{259--264}, \doi{10.1016/0012-365X(95)00120-L}.

\bibitemdeclare{inproceedings}{OllingerPolyo}
\bibitem{OllingerPolyo}
\bibinfo{author}{Nicolas \surnamestart Ollinger\surnameend}
  (\bibinfo{year}{2009}): \emph{\bibinfo{title}{{Tiling the Plane with a Fixed
  Number of Polyominoes}}}.
\newblock In: {\sl \bibinfo{booktitle}{Proceedings of the Third International
  Conference on Language and Automata Theory and Applications (LATA)}}, {\sl
  \bibinfo{series}{LNCS}} \bibinfo{volume}{5457}, pp.
  \bibinfo{pages}{638--647}, \doi{10.1007/978-3-642-00982-2\_54}.

\bibitemdeclare{article}{wangpatternrecoII}
\bibitem{wangpatternrecoII}
\bibinfo{author}{Hao \surnamestart Wang\surnameend} (\bibinfo{year}{1961}):
  \emph{\bibinfo{title}{{Proving theorems by Pattern Recognition II}}}.
\newblock {\sl \bibinfo{journal}{Bell Systems technical journal}}
  \bibinfo{volume}{40}, pp. \bibinfo{pages}{1--41}.

\end{thebibliography}
\end{document}